\documentclass[journal,12pt,onecolumn,draftclsnofoot]{IEEEtran}

\usepackage{graphicx}%
\usepackage{subcaption}
\usepackage{multirow}%
\usepackage{amsmath,amssymb,amsfonts}%
\usepackage{amsthm}%
\usepackage{mathrsfs}%
\usepackage{xcolor}%
\usepackage{textcomp}%
\usepackage{manyfoot}%
\usepackage{booktabs}%
\usepackage{algorithm}%
\usepackage{algorithmicx}%
\usepackage{algpseudocode}%
\usepackage{listings}%
\usepackage{makecell}
\usepackage{float}%
\usepackage{comment}
\theoremstyle{thmstyleone}%
\newtheorem{theorem}{Theorem}

\newtheorem{proposition}[theorem]{Proposition}%

\theoremstyle{thmstyletwo}%
\newtheorem{remark}{Remark}%

\theoremstyle{thmstylethree}%

\begin{document}

\title{Neilson's Weak vs. Strong Loss Aversion: A Characterization and a Generalized CPT-Utility Function}

\author{Symeon~Vaidanis~and~Marios~Kountouris
\thanks{S. Vaidanis is with the Communication Systems Department of EURECOM, Sophia Antipolis, 06410 France, and M. Kountouris is with the Communication Systems Department of EURECOM, Sophia Antipolis, 06410 France and Department of Computer Science and Artificial Intelligence of University of Granada, C/ Periodista Daniel Saucedo Aranda, Granada, 18071 Spain. E-mail: symeon.vaidanis@eurecom.fr, marios.kountouris@eurecom.fr, mariosk@ugr.es.}}

\maketitle

\begin{abstract}
In multi-objective and multi-criteria decision-making under risk, especially in settings involving individual behavior, risk-aware analysis based on subjective evaluation has become increasingly important. Moving beyond risk-neutral modeling and the constraints of Expected Utility Theory (EUT), Cumulative Prospect Theory (CPT) provides a behaviorally grounded framework for capturing how individuals perceive and evaluate risky prospects. This paper conducts a rigorous theoretical analysis of Neilson’s definitions of aversion. We provide a gamble-based interpretation, sharpen key conceptual distinctions, and make explicit the conditions under which the weak and strong notions coincide as well as when they diverge. Furthermore, we examine the K\"{o}bberling--Wakker utility function and related standard CPT specifications, highlighting \emph{structural limitations and inconsistencies} that arise when these forms are required to satisfy Neilson-type aversion conditions. To address these issues, we propose a generalized CPT-utility function that retains the canonical reference-dependent shape while offering additional flexibility. This generalization extends the descriptive scope of CPT and provides an explicit functional form that is useful for sensitivity analysis and utility function-based optimization.
\end{abstract}

\begin{IEEEkeywords}
Cumulative Prospect Theory, Loss Aversion, Reference-Dependent Preferences, Aversion Measures, Utility Function
\end{IEEEkeywords}

\section{Introduction}\label{Section1}
Decision-making under risk often requires comparing outcomes relative to a benchmark and assessing asymmetric responses to gains and losses. This paper focuses on that reference-dependent component of risky choice, with particular emphasis on loss-aversion conditions in Cumulative Prospect Theory (CPT). Rather than providing an empirical account of individual behavior, the analysis develops formal characterizations of Neilson-type aversion and examines the compatibility of standard CPT utility-function specifications with those characterizations.

Building on the foundational work of D. Bernoulli, J. von Neumann, O. Morgenstern, F. P. Ramsey, and L. J. Savage, among others, Expected Utility Theory (EUT) \cite{EUT} has profoundly shaped modern economic theory, statistical reasoning, and decision science by providing a normative framework for rational choice under uncertainty from a subjective perspective. Despite its elegance and conceptual appeal, EUT has faced significant empirical challenges, as numerous studies have shown that individuals often deviate systematically from its predictions. 

In response to these empirical limitations, Prospect Theory (PT) was introduced by \cite{Kahneman_Tversky_1979} as a descriptive model that accounts for key behavioral regularities, offering a more psychologically grounded perspective on how people make decisions under risk. PT extends EUT by introducing a reference-dependent framework that distinguishes gains from losses relative to a reference point. More formally, PT partitions the domain of the utility function into gains and losses, assumes risk aversion in the gain domain and risk seeking in the loss domain, and applies nonlinear transformations to outcomes and probabilities. Moreover, PT introduces probability weighting functions (PWFs), which distort objective probabilities to reflect the subjective perception of event occurrence. These functions capture the empirically observed tendencies of individuals to overweight low-probability events and underweight high-probability ones, patterns that traditional models fail to explain.
To address limitations of the original formulation, particularly in the treatment of prospects with multiple ranked outcomes, Cumulative Prospect Theory (CPT) was later developed by \cite{Kahneman_Tversky_1992} and further refined through the contributions of \cite{Kobberling_Wakker}, \cite{Prelec}, and \cite{Neilson_2002_A}. CPT extends PT by introducing cumulative probability weighting, thereby avoiding violations of first-order stochastic dominance.

As a result, CPT has become a central model in behavioral decision theory, with widespread applications in economics, finance, and psychology for explaining systematic deviations from rational choice. Although the present paper is theoretical in focus, CPT has also been applied in areas such as reinforcement learning \cite{CPT_RL}, control theory \cite{CPT_Control}, and telecommunications \cite{CPT_Alouini}. More recently, CPT has been used in the emerging field of goal-oriented semantic information theory \cite{kountouris2021semantics}, where concepts such as risk sensitivity and subjective outcome evaluation are used to characterize the semantic relevance and perceived value of information in telecommunication systems \cite{vaidanis2025ICC, vaidanis2025SPAWC}.

In this paper, we clarify foundational definitions and address persistent theoretical ambiguities in the literature, particularly within Neilson’s taxonomy of loss aversion \cite{Neilson_2002_A, Neilson_2007, Neilson_2017}. Our analysis shows that binary gambles serve not merely as illustrative tools, but as a tractable analytical foundation for characterizing and comparing core CPT aversion notions. Specifically, we establish formal links between Neilson's weak and strong forms of loss aversion, identifying the conditions under which these notions converge or diverge for generalized utility functions.
Furthermore, within both general utility theory and the CPT framework, we identify key limitations of the K\"{o}bberling--Wakker utility-function specification, including the reduction in amplitude as the parameters $\alpha,\beta$ increase and the constant absolute risk aversion (CARA) property, as well as limitations in standard CPT application settings, particularly the limited ability of standard parametric forms to capture asymmetric outcomes that are perceived as symmetric and to model behavior on bounded domains. To address these issues, we propose a broader and more flexible class of utility functions that generalizes the K\"{o}bberling--Wakker formulation while preserving behavioral consistency. This enhanced framework extends the class of admissible parametric CPT representations and yields a tractable parametric class for theoretical analysis. Taken together, these contributions enrich the theoretical foundations of CPT by introducing structurally grounded and behaviorally coherent extensions, and they strengthen the decision-theoretic interpretation of loss aversion under uncertainty.

In methodological terms, our analysis is conducted in a reference-dependent setting and focuses on the utility-function component of CPT. Although CPT also relies on probability weighting to model the subjective distortion of probabilities, the main results of this paper concern the characterization of loss aversion through properties of the utility function, with binary gambles serving as the primary analytical device. This focus allows us to isolate the role of local and global shape restrictions in the comparison of aversion notions and in the construction of tractable parametric specifications.

More precisely, we derive characterization results for Neilson’s weak and strong notions of loss aversion, including geometric and gamble-based interpretations, and we identify the additional tail-slope condition under which the two notions coincide for nonsmooth S-shaped utility functions. These results also clarify when the equivalence fails outside those conditions. We then use these theoretical insights to reassess standard parametric CPT-utility functions, highlighting restrictions that limit their ability to represent certain aversion patterns, and we propose a generalized K\"{o}bberling--Wakker-type class that preserves the canonical reference-dependent structure while allowing greater flexibility for theoretical analysis and comparative applications.

The paper is organized as follows. Section~\ref{Section2} reviews the CPT framework and fixes notation. Section~\ref{Section3} analyzes Neilson’s definitions of loss aversion and derives equivalence and non-equivalence results. Section~\ref{Section4} examines limitations of the K\"{o}bberling--Wakker utility-function specification and related standard CPT parametric forms, and then proposes an extended K\"{o}bberling--Wakker class that addresses these restrictions. Section~\ref{Section5} concludes.

\section{Cumulative Prospect Theory: Framework and Notation} \label{Section2}
In this section, we provide an overview of the mathematical framework underlying Cumulative Prospect Theory (CPT) \cite{Kahneman_Tversky_1992,WakkerTve93,ChaWak99}. Our goal is to introduce the key components of CPT, including utility functions, probability weighting functions, and decision weights. 
This section provides the notation and definitions used in the remainder of the paper and prepares the theoretical analysis in Sections~\ref{Section3} and \ref{Section4}.

CPT offers a compelling alternative to EUT and can be viewed as a generalization of EUT for modeling decision-making under risk. CPT provides a framework for evaluating choices among a finite set of prospects, represented as a collection of real-valued random variables $\mathcal{G}$. For each gamble $R \in \mathcal{G}$, the agent assigns a deterministic valuation $V(R)$ based on the CPT framework. The decision rule is to choose the prospect with the highest subjective valuation: $R^* = \arg\max_{R \in \mathcal{G}} V(R)$.

CPT introduces the following key features that depart from the assumptions of EUT: (i) \textit{reference dependence}, (ii) \textit{loss aversion}, and (iii) \textit{rank dependence and cumulative probability weighting}. Moreover, CPT incorporates the following two behavioral regularities:
\begin{itemize}
\item \textit{Diminishing sensitivity to outcomes}: sensitivity to changes in outcomes diminishes as the distance from the reference point increases, for both gains and losses. As we detail below, this implies that the utility function is concave for gains and convex for losses, reflecting decreasing marginal sensitivity in both domains. This property, grounded in a quantitative analysis of the utility-function component, is consistent with the probabilistic (prospect-based) analysis based on Jensen's inequality. In the original work of Kahneman and Tversky \cite{Kahneman_Tversky_1979}, this probabilistic (prospect-based) analysis was used to justify the concavity and convexity properties of the utility function on the gain and loss domains, respectively.
\item \textit{Probabilistic weighting}: individuals do not perceive and respond to objective probabilities in a linear manner. Instead, they tend to overweight small probabilities and underweight large ones. This behavior is captured by the probability weighting function in CPT, which reflects systematic distortions in the perception of probability.
\end{itemize}
It is important to note that \textit{diminishing sensitivity} is conceptually related to, but distinct from, \textit{loss aversion}: while loss aversion refers to the asymmetric valuation of losses relative to gains, diminishing sensitivity pertains to the curvature of the utility function within each domain. In the analysis developed later in the paper, we show that these concepts can also come into tension: under Neilson’s formulation, certain saturation features associated with diminishing sensitivity are incompatible with the corresponding aversion conditions. This point will be made precise in Section~\ref{Section3} through the characterization of aversion in terms of non-symmetric binary gambles. On the other hand, probabilistic sensitivity is directly tied to the probability weighting function, highlighting the departure from the linear treatment of probabilities assumed in EUT.

A decision maker, or more generally an agent, is characterized by a reference point $x_0 \in \mathbb{R}$, a corresponding utility function $u: \mathbb{R} \to \mathbb{R}$, and a pair of probability weighting functions $w^{+}, w^{-}: [0,1] \to [0,1]$, corresponding to gains and losses, respectively. These components jointly capture how the agent evaluates outcomes and perceives uncertainty. We refer to the triple $(x_0, u, w^{\pm})$ as the CPT specification of the agent.

\subsection{Reference Dependence and the CPT-utility function}
In CPT, outcomes are evaluated relative to a specified \emph{reference point} $x_0$. This reference point represents a benchmark level, such as an acquired or expected outcome, the status quo, an operating level, or a satisfaction level, against which gains and losses are evaluated. The choice of $x_0$ may vary across individuals, contexts, and application scenarios, depending on psychological or situational factors.

The domain of the utility function $u(x)$ is partitioned relative to $x_0$ into two regions: the loss domain $x < x_0$ and the gain domain $x \geq x_0$. In the loss domain, outcomes are assigned negative value, i.e., $u(x) < 0, \forall x < x_0$, and typically $\lim _{x \to x_0^{-}}u(x) \leq 0$. In the gain domain, outcomes are valued positively, i.e., $u(x)>0, \forall x > x_0$, and typically $\lim _{x \to x_0^{+}}u(x) \geq 0$. 
In general, the value at the reference point $u(x_0)$ can lie anywhere in the interval defined by $\lim _{x \to x_0^{-}}u(x) \leq u(x_0) \leq \lim _{x \to x_0^{+}}u(x)$, so the utility function may exhibit a jump discontinuity at $x_0$, i.e., $\lim_{x \to x_0^-} u(x) \neq \lim_{x \to x_0^+} u(x)$. This may make the utility function $u$ non-differentiable or non-smooth at $x_0$. However, in standard formulations of CPT, it is customary to normalize the reference point such that $\lim _{x \to x_0^{-}}u(x) = \lim _{x \to x_0^{+}}u(x) = u(x_0) = 0$, which simplifies the analysis and reflects the idea that the reference point represents a neutral benchmark, neither a gain nor a loss.

This reference-dependent evaluation is a key departure from EUT, where utility is typically defined over absolute outcomes and evaluated through a single utility index. In contrast, CPT’s distinction between gains and losses introduces systematic asymmetries in valuation, which are essential for capturing a wide range of empirically observed decision-making patterns under risk.

\subsection{Risk Sensitivity and Curvature of the Utility Function}
Assuming monotonicity, the utility function is taken to be strictly increasing over its entire domain, i.e., $u(x_1) < u(x_2), \; \forall x_1 < x_2$. As a result, the first derivative of the utility function satisfies:
\begin{equation}
\frac{\partial{u}}{\partial{x}} > 0, \quad \forall \; x \in \mathbb{R} \setminus \{x_0\}.
\end{equation}
The curvature of the utility function (marginal utility), reflected in its second derivative, characterizes the agent’s attitude toward risk (local risk attitude) and captures how the agent perceives changes of varying magnitudes within each subdomain (gains or losses).
In the gain domain ($x > x_0$):
\begin{itemize}
    \item A convex utility function indicates increasing marginal utility, meaning that the agent perceives larger gains (further from $x_0$) as increasingly valuable.
    \item A concave utility function indicates diminishing marginal utility, meaning that the agent is more sensitive to changes (gains) closer to $x_0$.   
\end{itemize} 
In the loss domain ($x < x_0$), the interpretation is reversed:
\begin{itemize}
    \item A convex utility function reflects diminishing sensitivity to larger losses (farther from $x_0$),
    \item whereas a concave function implies increasing sensitivity to such losses.
\end{itemize}
This differs from the standard assumption in EUT, where the utility function is typically assumed to be globally concave.

Furthermore, we consider a decision-maker choosing between an $N$-outcome gamble $X$ with probability distribution $\mathcal{P}$ and a sure (certain) outcome $x_1$, under the fair-gamble condition $\mathbb{E}_\mathcal{P}(X) = x_1$. The outcomes of the gamble lie within a single subdomain (i.e., entirely in the gain domain or entirely in the loss domain). In this setting:
\begin{itemize}
    \item A concave utility function represents a \emph{risk-averse} attitude, since the agent prefers the sure (certain) outcome to the gamble; by Jensen's inequality this is expressed as $\mathbb{E}_P(u(X))< u(x_1) = u(\mathbb{E}_P(X))$.
    \item A convex utility function represents \emph{risk-seeking} behavior, since the agent prefers the gamble to the sure outcome; by Jensen's inequality this is expressed as $\mathbb{E}_P(u(X)) > u(x_1) = u(\mathbb{E}_P(X))$.
\end{itemize}
If the agent’s preferences are invariant to the scale of outcomes, the corresponding utility function is linear, indicating risk neutrality. In contrast, a constant utility function models an agent who is completely insensitive to changes in outcomes, and thus expresses indifference across all values.
Within the classical CPT framework, the utility function is typically assumed to be S-shaped, satisfying:
\begin{align*}
\frac{\partial^2 u}{\partial x^2} &\leq 0 \quad \forall, x > x_0 \quad &\text{(concave in the gain domain)}, \\
\frac{\partial^2 u}{\partial x^2} &\geq 0 \quad \forall, x < x_0 \quad &\text{(convex in the loss domain)}.
\end{align*}
This structure captures the principle of diminishing marginal sensitivity: agents are more sensitive to changes close to $x_0$ and less sensitive as outcomes move further into the gain or loss regions.

\subsection{Certainty Equivalents and Risk Premia (Arrow–Pratt)}
In the framework of risk analysis developed by Arrow and Pratt \cite{Arrow,Pratt}, the main tools are the notions of certainty equivalent and risk premium. The certainty equivalent of a gamble $X$ with probability distribution $P$ is defined as $x_{CE} = u^{-1}\bigl(\mathbb{E}_{\mathcal{P}}[u(X)]\bigr)$, and can be interpreted as the sure amount that delivers the same level of utility as the expected utility of the gamble under the given utility function. The risk premium can be divided into two cases: the absolute risk premium and the relative risk premium, which are defined as follows:
\begin{subequations}
    \begin{equation}
         \pi_A = \mathbb{E}_P(X) - x_{CE}
    \end{equation}
    \begin{equation}
        \pi_R = \frac{\pi_A}{\mathbb{E}_P(X)} = 1 - \frac{x_{CE}}{\mathbb{E}_P(X)}
    \end{equation}
\end{subequations}
Here, $\pi_A$ denotes the absolute risk premium and $\pi_R$ denotes the relative risk premium, both expressed in terms of the certainty equivalent $x_{CE}$ and the expected value $\mathbb{E}_{\mathcal{P}}(X)$ of the gamble.

In their original papers, Arrow and Pratt derive local approximations for the risk premia by combining Taylor expansions around the certainty equivalent and around the mean of the gamble, and then taking expectations. In this way, the absolute and relative risk premia are obtained in terms of the local curvature of the utility function as follows:
\begin{subequations}
    \begin{equation}
        \pi_A \approx - \frac{1}{2} \cdot \frac{u''(\mathbb{E}_{\mathcal{P}}(X))}{u'(\mathbb{E}_P(X))} \cdot \sigma_X^2
        = \frac{1}{2} A(\mathbb{E}_{\mathcal{P}}(X)) \cdot \sigma_X^2
        \;,\;
        A(x) = - \frac{u''(x)}{u'(x)}
    \end{equation}
    \begin{equation}
        \pi_R \approx - \frac{1}{2} \cdot \frac{u''(\mathbb{E}_{\mathcal{P}}(X)) \cdot \mathbb{E}_{\mathcal{P}}(X)}{u'(\mathbb{E}_{\mathcal{P}}(X))} \cdot \frac{\sigma_X^2}{\mathbb{E}_{\mathcal{P}}^2(X)}
        = \frac{1}{2} R(\mathbb{E}_{\mathcal{P}}(X)) \cdot \frac{\sigma_X^2}{\mathbb{E}_{\mathcal{P}}^2(X)}
        \;,\;
        R(x) = - \frac{x\,u''(x)}{u'(x)}
    \end{equation}
\end{subequations}
where $\sigma_X^2$ is the variance (or, in this local approximation, the second central moment) of the gamble, and $A(x)$ and $R(x)$ are the Arrow--Pratt coefficients of absolute (ARA) and relative risk aversion (RRA), respectively. A utility function exhibits constant absolute risk aversion (CARA) if $A(x)$ is constant; equivalently, it is an affine transformation of an exponential function, i.e., $u(x)=\mu+\kappa \exp(\alpha x)$, in which case $A(x)=-\alpha$.

In the case of a fair decision problem between a certain outcome (stationary state) $x$ and a symmetric binary gamble around it $\{x-\delta,\frac{1}{2};x+\delta,\frac{1}{2}\}$, the absolute risk premium is defined as follows:
\begin{equation}
    u(x+\pi_A(x,\delta,u)) = \frac{1}{2}u(x+\delta) + \frac{1}{2}u(x-\delta).
\end{equation}
Following \cite{utility_premium_fundamental}, the utility premium for the previously defined decision setup is defined as
\begin{equation}
    \mathcal{U}(x,\delta) = u(x) - \left( \frac{1}{2}u(x+\delta) + \frac{1}{2}u(x-\delta) \right).
\end{equation}
Note that the component $\mathcal{E}(x,\delta) = \frac{1}{2}u(x+\delta) + \frac{1}{2}u(x-\delta)$ is the mean utility of the gamble.
Moreover, for the same decision setup, the probability premium $\gamma$ is the excess in winning probability over the fair condition that makes the agent indifferent between the stationary state and the binary gamble, \cite{microeconomic_theory}.
\begin{equation}
    u(x) = \left(\frac{1}{2} + \gamma(x,\delta)\right) u(x+\delta) + \left(\frac{1}{2} - \gamma(x,\delta)\right)u(x-\delta)
\end{equation}

\subsection{Loss Aversion: Definitions and Benchmark Conditions.}
The CPT-utility function exhibits \textit{loss aversion}, which refers to the phenomenon whereby agents are more sensitive to losses than to equivalent gains. The utility function is steeper for losses than for gains, reflecting a behavioral asymmetry in which individuals (agents) exhibit a stronger sensitivity to losses. 

Kahneman and Tversky introduced a related behavioral condition known as \emph{symmetric bet aversion} \cite{Kahneman_Tversky_1979}. This condition captures the idea that an agent strictly prefers not to engage in a symmetric gamble that offers equal magnitude gains and losses around the reference point. It is formally defined as:
\begin{equation}
    u(x_0 + \delta) + u(x_0 -\delta) < 0 \;,\; \forall \delta>0 \;,\; u(x_0)=0.
\end{equation} 
This condition implies that all symmetric fair gambles around the reference point are strictly rejected in favor of maintaining the status quo, reflecting a preference for certainty over balanced risk. By dividing each side by $\delta$ and taking the limit $\delta \to 0^+$, one obtains that the marginal utility for losses exceeds that for gains at the reference point. Formally, this is expressed as: $u'(x_0^+) < u'(x_0^-)$, where $u'(x_0^+)$ and $u'(x_0^-)$ denote the right- and left-hand derivatives of the utility function at the reference point $x_0$, respectively.

They also introduced a stronger version of this condition, known as \emph{increasing symmetric bet aversion}, which imposes an ordering on the aversion to symmetric bets of increasing magnitude. It is expressed as: 
\begin{equation}
    u(x_0 + \delta_1) + u(x_0 - \delta_1) < u(x_0 + \delta_2) + u(x_0 - \delta_2) \;,\; \forall 0 \leq \delta_2 < \delta_1.
\end{equation} 
This condition states that the disutility of a symmetric gamble becomes more pronounced as the stakes increase. Under suitable regularity assumptions (e.g., differentiability of $u$), it can be reformulated in differential terms as:
\begin{equation}
\frac{\partial{u}}{\partial{x}} \bigg\vert _{x = x_0 + \delta} < \frac{\partial{u}}{\partial{x}} \bigg\vert _{x = x_0 - \delta} \;,\; \forall \delta>0 \;,\; u(x_0) = 0.
\end{equation}
This inequality reinforces the asymmetry between gains and losses, emphasizing that marginal sensitivity to losses consistently exceeds that to gains for all deviations from the reference point.
Moreover, this condition implies that the degree of aversion to symmetric fair gambles increases monotonically with the magnitude of the deviation $\delta$, i.e., the rejection of such gambles becomes stronger as the absolute size of the gamble increases.

\cite{Neilson_2002_A} extended the standard definitions by introducing two refined concepts: \emph{Neilson's weak loss aversion} and \emph{Neilson's strong loss aversion}.
\emph{Neilson's weak loss aversion} is defined by the condition
\begin{equation}
    \frac{u(z)}{z-x_0} \leq \frac{u(y)}{y-x_0} \;,\; \forall y < x_0 < z,
\end{equation}
while \emph{Neilson's strong loss aversion} strengthens this notion by comparing local marginal utilities:
\begin{equation}
    \frac{\partial{u}}{\partial{x}} \bigg\vert _{x = z} \leq \frac{\partial{u}}{\partial{x}} \bigg\vert _{x = y} \;,\; \forall y < x_0 < z.
\end{equation} 
\emph{Neilson's weak loss aversion} implies the rejection of all non-symmetric, fair binary gambles around the reference point in favor of maintaining the status quo. \emph{Neilson's strong loss aversion} implies \emph{Neilson's weak loss aversion}. Under the regularity assumptions used below, the two notions coincide only when an additional no-saturation, or tail-slope, condition rules out flattening in the loss tail. Moreover, it is worth noting that most definitions of loss aversion in the literature (with the exception of \emph{Neilson's strong aversion}) are based on the quantitative intensity (utility-premium approach) of rejection of a specific fair gamble, as captured by Jensen's inequality. More precisely, the function $\mathbb{E}(x_0,\delta_1,\delta_2) = \frac{\delta_1}{\delta_1 + \delta_2} \cdot u(x_0 + \delta_2) + \frac{\delta_2}{\delta_1 + \delta_2} \cdot u(x_0 - \delta_1)$, which is the expected utility of a fair non-symmetric binary gamble, should have the following properties in each case:
\begin{itemize}
    \item \emph{Symmetric bet aversion}: $ \mathbb{E}(x_0,\delta,\delta) < u(x_0)$,
    \item \emph{Increasing symmetric bet aversion}: $ \frac{\partial \mathbb{E}(x_0,\delta,\delta) }{\partial \delta}< 0$,
    \item \emph{Neilson's weak loss aversion}: $\mathbb{E}(x_0,\delta_1,\delta_2) \leq u(x_0)$.
\end{itemize}

\subsection{Comparative Risk Sensitivity Measures}
A probabilistic risk-sensitivity framework was proposed in \cite{vaidanis2025Theoretical_Part_2_Probabilistic_Risk_Sensitivity}, offering several advantages over the utility-premium approach and providing a different perspective from the classical risk-premium and probability-premium framework \cite{Arrow, Pratt}. This new metric is more decision-making oriented than the metrics of the bibliography, which try to express the structure of fair decision-making setups under the effect of the utility function, and it captures the threshold front of the probability distribution, which determines the decision-making policy.  Although this framework is mathematically closely related to the probability premium, it is accompanied by a more rigorous graphical interpretation inspired by \cite{utility_premium_fundamental} than the graphical interpretation of the probability premium in \cite{probability_premium_graphical}, and it recovers classical loss-aversion definitions based on the utility premium, thereby yielding a more decision-oriented interpretation. In addition, the probabilistic risk sensitivity approach can be used to provide further insights into the risk behavior of an EUT utility function for binary gambles with unequal probability distributions.

Our framework is based on a geometric representation of the expected utility of a gamble, from which we derive a probability threshold for families of gambles with fixed outcome values. This probability threshold may refer to a choice problem between a gamble and a certain outcome, or between two gambles. Since this approach is introduced here in a binary-gamble setting, we restrict attention to binary gambles, and an extension to $N$-outcome gambles is left for future work. The main metric of risk sensitivity is the threshold ratio of the probability of gain to the probability of loss, $\mathcal{RS} \equiv \frac{P_{gain}}{P_{loss}}$, and the main definitions are the following:
\begin{itemize}
    \item The risk sensitivity metric for a symmetric binary gamble $\{x-\delta,P_{loss};x+\delta,P_{gain}\}$ around a stationary state (certain outcome) $x$ is
    \begin{equation}
        \mathcal{RS}_{stat}(x,\delta) = \frac{u(x) - u(x-\delta)}{u(x+\delta) - u(x)}
    \end{equation}
    The \textit{symmetric bet aversion} is defined as the case in which $\mathcal{RS}_{stat}(x,\delta)$ under the effect of the utility function is greater than in the risk neutral case (i.e., the case without the effect of the utility function), namely $\mathcal{RS}_{stat}(x,\delta) > 1$. The \textit{probabilistic increasing symmetric bet aversion} is defined as the case in which $\mathcal{RS}_{stat}(x,\delta)$ is an increasing function of $\delta$, i.e., $\frac{\partial \mathcal{RS}_{stat}(x,\delta)}{\partial \delta} > 0$.  
    \item The risk sensitivity metric for a non-symmetric binary gamble $\{x-w\delta_1,P_{loss};x+w\delta_2,P_{gain}\}$ (with $\delta_1$ and $\delta_2$ defining a fixed scaling direction and $w$ denoting the scaling factor) around a stationary state (certain outcome) $x$ is
    \begin{equation}
        \mathcal{RS}_{stat}(x,\delta_1,\delta_2,w) = \frac{u(x) - u(x-w\cdot\delta_1)}{u(x+w\cdot\delta_2) - u(x)} 
    \end{equation}
    The \textit{non-symmetric bet aversion} or \emph{Neilson's weak aversion} is defined as the case in which $\mathcal{RS}_{stat}(x,\delta_1,\delta_2,w)$ under the effect of the utility function is greater than in the risk neutral case (i.e., the case without the effect of the utility function), namely $\mathcal{RS}_{stat}(x,\delta_1,\delta_2,w) > \frac{\delta_1}{\delta_2}$. 
    \item The risk sensitivity metric for two nested binary gambles with the same probability distributions $\{x-w \cdot \delta_1,P_{loss};x+ w \cdot \delta_2,P_{gain}\}$, $\{x- w \cdot \delta_3,P_{loss};x+ w \cdot \delta_4,P_{gain}\}$ and $0 \leq \delta_1 \leq \delta_3$, $0 \leq \delta_2 \leq \delta_4$ (with $\delta_1,\delta_2,\delta_3,\delta_4$ defining a fixed scaling direction and $w$ denoting the scaling factor) is
    \begin{equation}
        \mathcal{RS}_{gamble}^{equal}(x,\delta_1,\delta_2,\delta_3,\delta_4) = \frac{u(x-w \cdot \delta_1) - u(x- w \cdot \delta_3)}{u(x+ w \cdot \delta_4) - u(x + w \cdot \delta_2)} 
    \end{equation}
    \item The risk sensitivity metric for two binary gambles with unequal probability distribution $A = \{x-\delta_1,P_{A,loss};x+\delta_2,P_{A,gain}\}$ and $B = \{x-\delta_3,P_{B,loss};x+\delta_4,P_{B,gain}\}$ with $A$ as the reference gamble, is
    \begin{equation}
        \mathcal{RS}_{gamble}^{unequal}(x,\delta_1,\delta_2,\delta_3,\delta_4, [P_A]) = \frac{\mathbb{E}[u(A)] - u(x-\delta_3)}{u(x+\delta_4) - \mathbb{E}[u(A)]}
    \end{equation}
    where $[P_A] = [P_{A,loss}, P_{A,gain}]$ and $\mathbb{E}[u(A)] = u(x-\delta_1) \cdot P_{A,loss} + u(x+\delta_2) \cdot P_{A,gain}$.
\end{itemize}
Furthermore, we note that if the utility function has a structure such that the probabilistic risk sensitivity $\mathcal{RS}$ is independent of the reference (certain) state $x$, then the utility function is called Constant Probabilistic Risk Sensitive (CPRS).

\subsection{Probability Weighting and Decision Weights}
A key feature of CPT is the non-linear distortion of probabilities, where objective probabilities are transformed through a PWF to reflect how individuals subjectively perceive uncertainty. The PWF captures the empirically observed tendency for individuals to overweight small probabilities and underweight moderate to high probabilities, deviating from the linear treatment in EUT. 

Formally, the probability weighting functions are denoted by $\omega^{\pm}: [0, 1] \to [0, 1]$, where $\omega^{+}$ is applied to gains and $\omega^{-}$ to losses. These functions are continuous, strictly increasing, and satisfy the boundary conditions: $\omega^{\pm}(0) = 0$ and $\omega^{\pm}(1) = 1$. 
Note that the functions $\omega^{+}$ and $\omega^{-}$ assign the same decision weight to the reference point outcome if and only if the following symmetry condition holds: $\omega^-(p) + \omega^+(1-p) = 1 \; \forall p \in [0, 1]$, which is equivalent to using a single, unified PWF for the entire cumulative distribution \cite{Ingersoll}. 
Moreover, a fundamental distinction between PT and CPT lies in the domain over which the PWFs are applied. In PT, the weighting is applied directly to the probability mass function (in discrete settings) or the probability density function (in continuous settings). This approach can violate first-order stochastic dominance. By contrast, CPT applies the weighting to cumulative probabilities, i.e., the cumulative distribution function (CDF) in the continuous case, preserving consistency with stochastic dominance and enabling a rank-dependent evaluation of outcomes.

One of the earliest and most influential forms of the PWF was proposed in \cite{Kahneman_Tversky_1992}, based on empirical and experimental evidence, and is defined as:
\begin{equation}
    \omega(p) = \frac{p^\delta}{\left( p^\delta + (1-p)^\delta \right)^{1/\delta}} \;,\; 0 < \delta \leq 1.
\end{equation}

The log-odds distortion function $\omega_{p_0,\gamma}$ proposed in \cite{LogOdd95} is defined by:
\begin{equation}
    \text{Lo}(w_{p_0,\gamma}(p)) = \text{Lo}(p) + (1 - \gamma)\, \text{Lo}(p_0), \quad \forall\, p \in (0, 1),
\end{equation}
where $\text{Lo}(p) := \ln\left( \frac{p}{1 - p} \right)$ denotes the log-odds transformation.

One of the most widely used forms of the probability weighting function is the Prelec function \cite{Prelec}, defined as:
\begin{equation}
    \omega(p) = \exp \left( -\gamma \left( -\ln(p) \right)^\theta \right)
    \label{eq:Prelec_PWF}
\end{equation}
where $p \in (0,1)$, $\theta > 0$, and $\gamma > 0$. 
The parameter $\theta$ controls the shape or curvature of the function, shaping the degree of probability distortion, while $\gamma$ determines the overall elevation of the curve, i.e., the location of the inflection point relative to the identity line $\omega(p) = p$, effectively shifting the function upward or downward.  For $0 < \theta < 1$ the function exhibits the familiar inverse-S pattern, with overweighting of small probabilities and underweighting of large probabilities. When $\theta > 1$, the weighting function becomes S-shaped instead. Table~\ref{tab:parameters_probability_weighting_function} summarizes the main shape properties for the parameter combinations considered in this paper.
\begin{table}[h]
    \centering
    \begin{tabular}{|c|c|c|c|}
        \hline
         & $0 < \gamma < 1$ & $\gamma = 1$ & $\gamma > 1$\\
        \hline
        $0 < \theta < 1$ & \makecell{inverse S-shape, \\ $\tilde{p} < \omega(\tilde{p})$}  & \makecell{inverse S-shape, \\ $\tilde{p} = \omega(\tilde{p})$} & \makecell{inverse S-shape, \\ $\tilde{p} > \omega(\tilde{p})$} \\
        \hline
        $\theta = 1$ & \makecell{strictly concave, \\ $\tilde{p} < \omega(\tilde{p})$} & \makecell{linear, \\ $\tilde{p} = \omega(\tilde{p})$} & \makecell{strictly convex, \\ $\tilde{p} > \omega(\tilde{p})$} \\
        \hline
        $\theta > 1$ & \makecell{S-shape, \\ $\tilde{p} < \omega(\tilde{p})$} & \makecell{S-shape, \\ $\tilde{p} = \omega(\tilde{p})$} & \makecell{S-shape, \\ $\tilde{p} > \omega(\tilde{p})$} \\
        \hline
    \end{tabular}
    \caption{Shape properties of the Prelec PWF for different parameter regimes. The inflection point is denoted by $\tilde{p}$.}
    \label{tab:parameters_probability_weighting_function}
\end{table}

Taken together, the parameter combinations in Table~\ref{tab:parameters_probability_weighting_function} produce both inverse-S and S-shaped weighting patterns. In particular, the parameter $\theta$ determines whether the function is inverse-S ($0 < \theta < 1$) or S-shaped ($\theta > 1$), while $\gamma$ shifts the curve relative to the identity line $\omega(p)=p$. Accordingly, the pattern of overweighting and underweighting depends on the parameter regime: the inverse-S case ($0 < \theta < 1$) captures the standard CPT pattern of overweighting small probabilities and underweighting large probabilities, whereas other regimes may generate different distortions.

\subsection{Common Parametric CPT-utility functions}
Two widely used utility functions within the CPT framework have been proposed in the literature. The first is the Kahneman and Tversky utility function \cite{Kahneman_Tversky_1992}, defined as follows:
\begin{equation}
    u(x) = \left\{
        \begin{array}{ll}
            \left(x - x_0 \right)^\alpha & \textrm{for} ~x \geq x_0 \\
            - \lambda\left(x_0 - x \right)^\beta & \textrm{for} ~x < x_0 \\
        \end{array} 
    \right. 
    \label{eq:KT_utility}
\end{equation}
where $\alpha, \beta \in  (0, 1]$ (for an S-shaped utility function) govern the curvature of the utility function and capture diminishing sensitivity to gains and losses, respectively, while $\lambda >0$ represents the degree of loss aversion. Despite its intuitive appeal and empirical relevance, this formulation has two notable limitations: (i) it is generally non-smooth at $x_0$: for $0<\alpha,\beta<1$, $\displaystyle \lim_{x \to x_0^-} \frac{\partial u}{\partial x}=+\infty$ and $\displaystyle \lim_{x \to x_0^+} \frac{\partial u}{\partial x}=+\infty$, whereas for boundary cases $\alpha=1$ or $\beta=1$ the corresponding one-sided derivative is finite; and (ii) symmetric bet aversion (or loss aversion under symmetric gambles) is satisfied only if $\alpha = \beta$ and $\lambda > 1$ \cite{Al-Nowaihi}.

Another widely used utility function, proposed by K\"{o}bberling and Wakker \cite{Kobberling_Wakker}, is given by:
\begin{equation}
    u(x) = \left\{
        \begin{array}{ll}
            \dfrac{1 - \exp\left( - \alpha (x - x_0)  \right)}{\alpha} & \text{for}~ x_0 \leq x \\
            - \lambda \dfrac{1 - \exp\left( - \beta (x_0 - x)  \right)}{\beta} & \text{for}~ x < x_0 \\
        \end{array} 
    \right. 
    \label{eq:KW_utility}
\end{equation}
where $\lambda$, $\alpha$, and $\beta$ are agent-specific parameters. The positivity of all parameters ensures that the function is strictly increasing and exhibits the S-shaped form characteristic of CPT.
This formulation allows for greater flexibility than the classical power function, particularly in controlling curvature and asymmetry. Moreover, the condition for increasing symmetric bet aversion can be satisfied if $\alpha > \beta$ and $\lambda > 1$ \cite{Dhami_book}.

\subsection{CPT Evaluation Functional (Discrete and Continuous Prospects)}
The subjective value attributed to a prospect $R$, when $R$ is a discrete random variable, is given by
\begin{equation}
    V(R) = \sum_{i=-n}^{m} \Omega_i \cdot u(r_i),
\end{equation}
where the cumulative probability weighting under the ordering $r_{-n} \leq r_{-n+1} \leq \dots r_{0} \leq \dots \leq r_{m-1} \leq r_{m}$ is defined by 
\begin{subequations}
    \begin{equation}
        P_{-i} = p(R = r_{-n}) + \dots + p(R = r_{-i}) \rightarrow \Omega_{-i} = \omega^-(P_{-i}) - \omega^-(P_{-i-1})
    \end{equation}
    \begin{equation}
        P_{i}^c = p(R = r_{i}) + \dots + p(R = r_{m}) \rightarrow \Omega_{i} = \omega^+(P_{i}^c) - \omega^+(P_{i+1}^c)
    \end{equation}
\end{subequations}
Under the symmetry condition $\omega^-(p) + \omega^+(1 - p) = 1$, one may equivalently work with a single unified weighting function \cite{Ingersoll}.

Regarding the case of a continuous prospect $R$, the subjective value attributed to $R$ is given by
\begin{equation}
    V(R) = \int_{-\infty}^{x_0} u(r)\, d[\omega^-(F_R(r))] - \int_{x_0}^{+\infty} u(r)\, d[\omega^+(\bar{F}_R(r))],
\end{equation}
where $\omega^-,\omega^+$ are the probability weighting functions for the loss and gain domains, $F_R(r) = \mathbb{P}(R \leq r)$ is the CDF of the prospect $R$, and $\bar{F}_R(r) = \mathbb{P}(R > r) = 1 - F_R(r)$ is the survival (tail) function. Under the symmetry condition discussed above (see \cite{Ingersoll}), one may use a single unified probability weighting function $\omega$, in which case the CPT valuation can be written as
\begin{equation}
    V(R) = \int_{-\infty}^{+\infty} u(r)\, f_R(r) \, \frac{d\omega(F_R(r))}{dF_R(r)} \, dr, 
\end{equation}
where $\omega$ is the common probability weighting function for both subdomains and $f_R(r)$ is the PDF of the prospect $R$.

\section{Neilson's Weak and Strong Loss Aversion: Definitions and Characterizations}\label{Section3}
While loss aversion has been extensively examined through the work of Kahneman and Tversky, particularly within the framework of PT, Neilson’s alternative formulation offers a theoretically rich and comparatively underexplored perspective. Despite its potential to provide deeper insights into behavioral asymmetries in risk evaluation, Neilson’s approach has received relatively limited attention in both theoretical development and empirical application.

This section provides a detailed analysis and interpretation of Neilson’s definitions of aversion. Our goal is to clarify their conceptual foundations, evaluate how they align with or diverge from more widely studied models, and strengthen their theoretical underpinnings. We focus in particular on characterization results, the relationship between weak and strong loss aversion, and the role of nonsmooth, S-shaped utility functions in establishing equivalence.

\subsection{Neilson's Weak Loss Aversion and Non-symmetric Bet Aversion}
In our terminology, Neilson's concept of \emph{weak loss aversion} is more naturally interpreted as \emph{non-symmetric bet aversion}, since it entails the rejection of all \textit{fair but non-symmetric} binary gambles around a reference point. This characterization highlights the structural asymmetry in outcome sensitivity rather than a broad notion of aversion to losses per se.

It is important to note that \emph{Neilson’s weak loss aversion} does \textit{not} imply \emph{increasing symmetric bet aversion} in general. That is, even if an individual rejects all non-symmetric fair gambles around a reference point, this does not necessarily mean that the degree of aversion increases with the scale of symmetric gambles. An intuitive counterexample illustrating this separation is provided in Figure~\ref{fig:Neilson_weak_but_not_strong}.

\begin{figure}[H]
    \centering
    \includegraphics[width=0.5\linewidth]{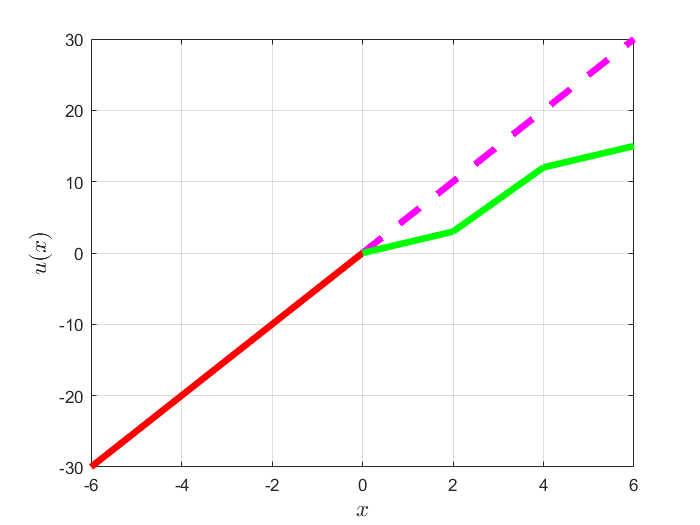}
    \caption{Schematic counterexample: a utility function satisfying Neilson's weak loss aversion while violating the marginal-slope ordering required for Neilson's strong loss aversion. The violation occurs for pairs $y<x_0<z$ such that $u'(z)>u'(y)$.}
    \label{fig:Neilson_weak_but_not_strong}
\end{figure}

In contrast, \emph{Neilson's strong loss aversion} does imply \emph{increasing symmetric bet aversion}, a result that follows directly from the formal definitions. This observation reveals a conceptual and behavioral \textit{gap} between the \emph{weak} and \emph{strong} forms of \emph{Neilson's loss aversion}.

To bridge this gap, we introduce the notion of \textbf{increasing non-symmetric bet aversion}, formulated within a utility-premium framework that quantifies the intensity of gamble rejection. This intermediate concept captures individuals who not only reject all non-symmetric fair gambles but also exhibit increasing aversion as the stakes of such gambles scale. Importantly, this condition:
\begin{itemize}
    \item Implies \emph{increasing symmetric bet aversion},
    \item Is strictly stronger than \emph{Neilson’s weak loss aversion}, and
    \item Constitutes a special case within \emph{Neilson’s strong loss aversion}.
\end{itemize}

\begin{theorem}\label{increasing_non_symmetric}
Let $u: \mathbb{R} \to \mathbb{R}$ be differentiable on \((-\infty,x_0)\cup(x_0,\infty)\), with \(u(x_0)\) well defined. Fix \(\delta_1,\delta_2>0\), and consider the family of fair non-symmetric binary gambles around \(x_0\), scaled by \(w>0\), with outcomes
\[
x_0-w\delta_1
\qquad\text{and}\qquad
x_0+w\delta_2 .
\]
Then the utility-premium criterion for these scaled gambles is strictly increasing in \(w\) if and only if
\[
u'(x_0+w\delta_2) < u'(x_0-w\delta_1), \qquad \forall\, w>0.
\]
Equivalently, for \(p=\dfrac{\delta_1}{\delta_1+\delta_2}\), the expected utility of the scaled gamble is strictly decreasing in \(w\), i.e.,
\[
p\,u(x_0+w_1\delta_2)+(1-p)\,u(x_0-w_1\delta_1)
>
p\,u(x_0+w_2\delta_2)+(1-p)\,u(x_0-w_2\delta_1),
\]
for all scaling factors \(w_2>w_1\geq0\).
\end{theorem}

\begin{proof}
See Appendix A.
\end{proof}

This result provides a clean differential characterization of increasing non-symmetric bet aversion in terms of the marginal utility gradient around the reference point. It shows that aversion increases with the scale of the gamble when marginal sensitivity to losses consistently exceeds marginal sensitivity to gains along the corresponding scaled outcomes.

\subsection{Neilson's Weak Loss Aversion: Tangent-Line (Geometric) Characterization}
To gain deeper insight into \emph{Neilson’s} notion of \emph{weak loss aversion}, we provide an equivalent formulation that lends itself to a more intuitive geometric interpretation.

\begin{theorem} \label{weak_loss_aversion_equivalent}
Let $u: \mathbb{R} \to \mathbb{R}$ be a strictly increasing, twice continuously differentiable with an exception at the reference point, where it is nonsmooth, and S-shaped utility function such that $u(x_0) = 0$. Then, the function $u$ satisfies Neilson’s \emph{weak aversion} if and only if:
\[
u(x) \leq u'(x_0^+) \cdot (x - x_0), \quad \forall x < x_0.
\]
\end{theorem}
\begin{proof}
    See Appendix A.
\end{proof}

Thus, \emph{Neilson's weak loss aversion} is equivalent to requiring that the loss branch of the utility function lie strictly below (or on) the right tangent at the reference point.

This result reveals that \emph{Neilson's weak loss aversion} is equivalent to the condition that, for all losses $x < x_0$, the utility function lies strictly below the tangent line to the gain domain at $x_0$, extended into the loss domain. This provides a clear graphical criterion: the slope of the utility function just to the right of the reference point defines an upper linear bound on the actual value assigned to losses.

As a consequence, we can reformulate \emph{Neilson’s weak loss aversion} in the following way:
\begin{equation} 
 \label{eqn:Reformulate_Neilson}
    u'(x_0^-) \cdot (x - x_0) \leq u(x) \leq u'(x_0^+) \cdot (x - x_0), \quad \forall x < x_0
\end{equation}
where
\begin{itemize}
    \item The left inequality follows under the assumption that the utility function is convex in the loss domain (a common CPT specification), together with Jensen’s inequality.
    \item The right inequality is the defining condition of Neilson’s weak loss aversion.
\end{itemize}
This formulation allows us to make an important observation about the class of utility functions that satisfy \textit{Neilson’s weak loss aversion}. In particular, for the inequality $u(x) \leq u'(x_0^+) \cdot (x - x_0)$ to hold as $x \to -\infty$, the following must be true:
\[
0 \leq \lim_{x \to -\infty} \left( u'(x_0^+) \cdot (x - x_0) - u(x) \right)
\quad \text{and} \quad
\lim_{x \to -\infty} \left( u'(x_0^+) \cdot (x - x_0) \right) = -\infty.
\]
These conditions imply that the utility function must not exhibit a saturation effect in the limit as $x \to -\infty$. That is:
\[
\lim_{x \to -\infty} u(x) \neq c \in \mathbb{R}^- \quad \text{and} \quad \lim_{x \to -\infty} u'(x) \neq 0.
\]
In contrast, this behavior is in tension with the principle of diminishing marginal sensitivity in the loss domain, as proposed by Kahneman and Tversky in their original formulation of PT, where the utility function tends to flatten for extreme losses.

Hence, \textit{Neilson’s weak loss aversion} fundamentally differs from Kahneman and Tversky’s framework in its treatment of deep losses: it requires the utility function to continue decreasing at a sustained, non-negligible rate as losses grow, whereas Kahneman and Tversky allow the utility function to asymptotically flatten, possibly approaching a finite lower bound.

\subsection{Weak–Strong Relations under S-Shaped Preferences}
As previously noted, Neilson claims an equivalence between his definitions of \emph{Neilson's weak loss aversion} and \emph{Neilson's strong loss aversion}, under the assumption that the utility function is strictly increasing, twice continuously differentiable, and S-shaped. While the implication from \textit{Neilson's strong} to \textit{weak loss aversion} is correctly established in his proof, the converse direction, from weak to strong, is incomplete and lacks rigor.

In this subsection, we address this gap by providing a complete and corrected proof of the reverse implication. Specifically, we show that under S-shaped preferences, \emph{Neilson's weak loss aversion} implies \emph{Neilson's strong loss aversion} provided a mild no-saturation condition holds in the loss tail.

\begin{theorem}\label{Neilson_S_shape_equivalence}
Let $u:\mathbb{R} \to \mathbb{R}$ be a strictly increasing, twice continuously differentiable with an exception at the reference point, where it is nonsmooth, and S-shaped utility function such that $u(x_0) = 0$. Assume that $u$ satisfies Neilson's weak loss aversion and, in addition, the loss-domain marginal value satisfies the \emph{no-saturation (tail-slope) condition} 
\begin{equation}\label{eq:NS_tail_slope}
\inf_{x<x_0} u'(x)\;\ge\;u'(x_0^+).
\end{equation} Then $u$ satisfies Neilson’s strong loss aversion.
\end{theorem}
\begin{proof}
See Appendix A.
\end{proof}

\begin{remark}
Condition \eqref{eq:NS_tail_slope} rules out loss-domain saturation and makes explicit the asymptotic requirement needed to upgrade Neilson’s weak loss aversion to strong loss aversion, thereby ensuring the pointwise marginal-utility ordering. In many standard CPT specifications that flatten for extreme losses, condition \eqref{eq:NS_tail_slope} fails. On unbounded domains, such loss-tail saturation also violates Neilson-type weak and strong loss-aversion conditions.
\end{remark}
This result identifies the additional tail-slope restriction needed to guarantee the equivalence between Neilson’s weak and strong loss aversion under nonsmooth S-shaped preferences. Under this no-saturation condition, weak loss aversion can be upgraded to the marginal-utility ordering required for strong loss aversion.

\subsection{Neilson's Strong Loss Aversion: Average-Slope Characterization and Gamble Interpretation} \label{sec:strong_neilson_probabilistic}
To better understand the conceptual and operational meaning of \emph{Neilson’s strong loss aversion}, we now present an equivalent formulation in terms of average slopes (equivalently, average marginal utilities) across gains and losses. This alternative expression facilitates comparison, connects naturally with the subcases discussed earlier, and leads to a natural gamble-based interpretation.

\subsubsection{An Average-Slope Inequality for Strong Loss Aversion}
\begin{theorem} \label{Neilson_strong_equivalent}
Let $u:\mathbb{R} \to \mathbb{R}$ be strictly increasing and differentiable on
$(-\infty,x_0)$ and $(x_0,\infty)$, with $u(x_0)=0$\footnote{Allowing $\delta_3=0$ or $\delta_4=0$ makes $x_0$ an endpoint in the Mean Value Theorem step; we therefore assume $u$ is continuous at $x_0$.}. Then the following are equivalent:
\begin{enumerate}
\item[(i)] (\emph{Non-strict strong loss aversion}) For all $y<x_0<z$,
\[
u'(z)\le u'(y).
\]
\item[(ii)] (\emph{Average-slope inequality}) For all $\delta_1,\delta_2,\delta_3,\delta_4\ge 0$
such that $\delta_1\neq \delta_3$, $\delta_2\neq \delta_4$, and $\delta_1+\delta_3>0$, $\delta_2+\delta_4>0$,
\[
\frac{u(x_0+\delta_2)-u(x_0+\delta_4)}{\delta_2-\delta_4}
\le
\frac{u(x_0-\delta_3)-u(x_0-\delta_1)}{\delta_1-\delta_3}.
\]
\end{enumerate}
\end{theorem}
\begin{proof}
    See Appendix A.
\end{proof}
This inequality states that the average slope of the utility function over gain intervals is no larger than the corresponding average slope over loss intervals. In this form, Neilson’s condition becomes directly comparable across intervals and is especially convenient for deriving gamble-based interpretations. More specifically, the inequality compares average marginal utilities over distinct intervals in the gain and loss domains, and it captures the core implication of strong loss aversion: losses carry greater marginal impact than comparable gains, even when the associated intervals differ in magnitude.

This general condition subsumes several previously defined forms of aversion. By selecting specific parameter configurations or taking appropriate limits, we recover notable subcases of Neilson’s framework:
\begin{itemize}
    \item \textbf{Symmetric Bet Aversion:}  
    Setting $\delta_4 = \delta_1 = 0$ and $\delta_2 = \delta_3 = \delta$ yields the standard case of symmetric bet aversion, where a fair symmetric gamble is weakly rejected.

    \item \textbf{Increasing Symmetric Bet Aversion:}  
    Taking the limits $\delta_2 \to \delta_4$ from below and $\delta_3 \to \delta_1$ from above, and then setting $\delta_1 = \delta_4 = \delta$, leads to increasing symmetric bet aversion, i.e., increasing rejection of symmetric bets as the stake size increases.

    \item \textbf{Neilson's Weak Loss Aversion (Non-Symmetric Fair Bet Rejection):}  
    Letting $\delta_1 = \delta_4 = 0$ also recovers the condition under which all fair non-symmetric binary gambles are rejected (Neilson's weak loss aversion).

    \item \textbf{Increasing Non-Symmetric Bet Aversion:}  
    Finally, taking the limits $\delta_2 \to \delta_4$ from below and $\delta_3 \to \delta_1$ from above and then scaling $\delta_1 = w \cdot \tilde{\delta}_1$ and $\delta_4 = w \cdot \tilde{\delta}_4$ captures increasing non-symmetric bet aversion, that is, increasing rejection of non-symmetric gambles as their common scale $w$ grows.
\end{itemize}
We emphasize that the term ``increasing rejection'' belongs to the utility-premium approach, where it refers to the intensity of rejection of a gamble. In summary, this equivalence and its derived special cases provide a unified framework for interpreting and analyzing various forms of loss aversion within a common inequality-based structure.

\subsubsection{Gamble-based Interpretation of Neilson's Strong Loss Aversion}
\emph{Neilson’s} notion of \emph{strong loss aversion} captures risk sensitivity \emph{between} binary gambles, rather than only relative to a single certain outcome. As such, it can be interpreted and recovered as an inequality involving the risk-sensitivity measure between binary gambles with equal probability distributions. Specifically, an equivalent condition for \emph{Neilson’s strong loss aversion} can be expressed using the risk-sensitivity index $\mathcal{RS}_{\text{gamble}}^{\text{equal}}$ as follows:

\begin{itemize}
    \item In the \textit{trivial ordering cases}, where either $\delta_3 < \delta_1$ and $\delta_4 < \delta_2$, or $\delta_1 < \delta_3$ and $\delta_2 < \delta_4$, the inequality
    \[
        0 > \frac{(x_0 - \delta_3) - (x_0 - \delta_1)}{(x_0 + \delta_4) - (x_0 + \delta_2)}
        >
        \mathcal{RS}_{\text{gamble}}^{\text{equal}}(x_0, \delta_1, \delta_2, \delta_3, \delta_4)
    \]
    captures the dominance of the outer gamble’s gain--loss structure over that of the inner gamble.

    \item In the \textit{nested cases}, where either $\delta_3 < \delta_1$ and $\delta_2 < \delta_4$, or $\delta_1 < \delta_3$ and $\delta_4 < \delta_2$, the inequality reverses:
    \[
        0 < \frac{(x_0 - \delta_3) - (x_0 - \delta_1)}{(x_0 + \delta_4) - (x_0 + \delta_2)}
        <
        \mathcal{RS}_{\text{gamble}}^{\text{equal}}(x_0, \delta_1, \delta_2, \delta_3, \delta_4),
    \]
    indicating that, under strong loss aversion, the inner gamble (with the tighter outcome range) is preferred to the outer one.
\end{itemize}

Hence, the interpretation is as follows:
\begin{itemize}
    \item From a subjective perspective, in the case of \emph{nested gambles}, the preference region shifts in favor of the nested (inner) gamble relative to the risk-neutral case, and therefore the region in which the non-nested (outer) gamble is preferred shrinks correspondingly.
    
    \item In the case of \emph{non-nested gambles}, the shape of the utility function tends to increase the preference for the gamble that is more shifted toward gains, reflecting a bias in favor of positively skewed outcomes.
\end{itemize}
Thus, the sign and ordering of \(\mathcal{RS}_{\text{gamble}}^{\text{equal}}\) provide a direct gamble-based test for strong loss aversion across different geometric configurations of binary gambles.

\subsection{Neilson Loss Aversion and Rothschild-Stiglitz Risk Aversion under EUT}
\label{subsec:neilson_rs_eut}
The average-slope characterization in Section~\ref{sec:strong_neilson_probabilistic} also clarifies the connection between Neilson's strong loss aversion and classical Rothschild-Stiglitz risk aversion \cite{Risk_aversion_Rothschild_Stiglitz} under EUT. The connection is not an equivalence between the two concepts. Rather, both can be expressed through comparisons of average marginal utility over intervals. Rothschild-Stiglitz risk aversion concerns global aversion to mean-preserving spreads of outcomes under EUT, whereas Neilson's strong loss aversion concerns a reference-dependent comparison between marginal valuations in the loss and gain domains.

Consider a finite prospect
\[
    L=\{(x_1,p_1),\ldots,(x_n,p_n)\},
\]
where $x_i$ are outcomes, $p_i>0$, and $\sum_{i=1}^{n}p_i=1$. Let $x_i>x_j$. A two-point mean-preserving contraction transfers expected value from the higher outcome $x_i$ to the lower outcome $x_j$. For $\delta>0$, define
\[
    \widetilde x_i=x_i-\frac{\delta}{p_i},\qquad
    \widetilde x_j=x_j+\frac{\delta}{p_j},
\]
with all other outcomes unchanged. Since
\[
    p_i\left(x_i-\frac{\delta}{p_i}\right)
    +
    p_j\left(x_j+\frac{\delta}{p_j}\right)
    =
    p_i x_i+p_j x_j,
\]
the transformation preserves the mean. The modified prospect is
\[
    \widetilde L
    =
    \{(x_1,p_1),\ldots,(\widetilde x_i,p_i),\ldots,
    (\widetilde x_j,p_j),\ldots,(x_n,p_n)\}.
\]
A risk-averse EUT decision maker prefers such mean-preserving contractions, or equivalently dislikes the corresponding mean-preserving spreads.

Under EUT, the condition $\widetilde L\succeq L$ is equivalent to
\[
    p_i u\!\left(x_i-\frac{\delta}{p_i}\right)
    +
    p_j u\!\left(x_j+\frac{\delta}{p_j}\right)
    \geq
    p_i u(x_i)+p_j u(x_j).
\]
Rearranging gives
\[
    p_j\left[
    u\!\left(x_j+\frac{\delta}{p_j}\right)-u(x_j)
    \right]
    \geq
    p_i\left[
    u(x_i)-u\!\left(x_i-\frac{\delta}{p_i}\right)
    \right].
\]
Equivalently, using $\Delta_i=\delta/p_i$ and $\Delta_j=\delta/p_j$, this becomes
\[
    \frac{u(x_j+\Delta_j)-u(x_j)}{\Delta_j}
    \geq
    \frac{u(x_i)-u(x_i-\Delta_i)}{\Delta_i}.
    \label{eq:rs_average_slope_condition}
\]
Thus, a two-point mean-preserving contraction is preferred under EUT precisely when the average marginal utility over the lower-outcome interval is at least as large as the average marginal utility over the higher-outcome interval. If this condition holds for all such pairs of intervals, it is the standard average-slope characterization of concavity, and therefore of Rothschild-Stiglitz risk aversion under EUT.

The connection with Neilson's strong loss aversion is obtained by placing the two intervals on opposite sides of a reference point $x_0$. Suppose that the lower interval lies in the loss domain and the upper interval lies in the gain domain. Set
\[
    x_j=x_0-\delta_1,\qquad
    x_j+\Delta_j=x_0-\delta_3,
\]
and
\[
    x_i-\Delta_i=x_0+\delta_2,\qquad
    x_i=x_0+\delta_4,
\]
where
\[
    0\leq \delta_3<\delta_1,\qquad
    0\leq \delta_2<\delta_4.
\]
Then \eqref{eq:rs_average_slope_condition} becomes
\[
    \frac{
    u(x_0-\delta_3)-u(x_0-\delta_1)
    }{
    \delta_1-\delta_3
    }
    \geq
    \frac{
    u(x_0+\delta_4)-u(x_0+\delta_2)
    }{
    \delta_4-\delta_2
    }.
    \label{eq:neilson_rs_cross_domain}
\]
This is exactly the average-slope form of Neilson's strong loss aversion derived in Section~\ref{sec:strong_neilson_probabilistic}. Hence, Neilson's strong loss aversion can be interpreted as a reference-dependent, cross-domain analogue of the Rothschild-Stiglitz average-slope condition: losses below the reference point must have average marginal impact at least as large as gains above the reference point.

This interpretation also clarifies the difference between the two concepts. Rothschild-Stiglitz risk aversion under EUT is a global property: a decision maker prefers every mean-preserving contraction if and only if the Bernoulli utility function is concave. Neilson's strong loss aversion is instead a cross-domain property: it compares average marginal utility in the loss domain with average marginal utility in the gain domain. Therefore, Neilson's condition does not require global concavity of the utility function. This distinction is essential in CPT, where the utility function is typically convex over losses and concave over gains. In that setting, Neilson's condition imposes a dominance of loss-domain marginal valuation over gain-domain marginal valuation, rather than ordinary global risk aversion.

A second ordering case arises when the two-point transfer is large enough that the modified outcomes cross. Suppose
\[
    x_j < x_i-\Delta_i < x_j+\Delta_j < x_i.
\]
Let
\[
    a=x_j,\qquad b=x_i-\Delta_i,\qquad
    c=x_j+\Delta_j,\qquad d=x_i,
\]
so that $a<b<c<d$. The EUT preference $\widetilde L\succeq L$ becomes
\[
    \frac{u(c)-u(a)}{c-a}
    \geq
    \frac{u(d)-u(b)}{d-b}.
    \label{eq:rs_crossing_intervals}
\]
For a globally concave utility function, this inequality follows from the fact that secant slopes are nonincreasing as the interval moves to the right. Thus, the crossing case is also consistent with Rothschild-Stiglitz risk aversion under EUT. However, it is less directly tied to Neilson's loss-aversion interpretation, because each interval may span both sides of the reference point and therefore does not isolate a pure loss-domain interval versus a pure gain-domain interval.

The risk-seeking analogue is obtained by reversing the transfer. If $\delta<0$, then value is transferred from the lower outcome to the higher outcome, producing a two-point mean-preserving spread rather than a contraction. Preference for this spread under EUT is equivalent to the reverse average-slope inequality and is characterized globally by convexity of the utility function. This is the opposite of Rothschild-Stiglitz risk aversion and should therefore be interpreted as risk seeking, not as risk aversion.

In summary, the Rothschild-Stiglitz and Neilson frameworks are linked through the same mathematical object: average marginal utility over intervals. Under EUT, imposing the average-slope ordering globally yields concavity and, therefore, aversion to mean-preserving spreads. In the reference-dependent CPT setting, imposing the corresponding ordering across the loss and gain domains yields Neilson's strong loss aversion. The former is a global risk-aversion condition, whereas the latter is a reference-dependent loss-aversion condition.

Hence, under the EUT framework, if a two-point mean-preserving modification produces an outer gamble that is rejected in favor of the corresponding inner gamble, then the utility function satisfies the concavity-type average-slope ordering over the relevant interval. Conversely, if the modification produces an inner gamble that is rejected in favor of the corresponding outer gamble, then the utility function satisfies the convexity-type average-slope ordering over the relevant interval.

\section{Parametric CPT-utility functions: Limitations and a Generalization} \label{Section4}

In this section, we analyze structural limitations of the K\"{o}bberling--Wakker value-function specification under both general utility-theoretic considerations and standard CPT application scenarios. We highlight the resulting consistency restrictions and then introduce a generalized K\"{o}bberling--Wakker class that resolves these issues while preserving behavioral coherence.

\subsection{Limitations of the K\"{o}bberling--Wakker Specification}
A first limitation of the K\"{o}bberling--Wakker specification concerns the role of the curvature parameters $\alpha$ and $\beta$ in the induced probabilistic risk-sensitivity measures. These coefficients have a critical impact on probabilistic risk sensitivity in cases where the outcome values lie entirely within a single subdomain.
\begin{subequations}
    \begin{equation}
    \begin{split}
        (x,\delta) : 0 \leq \delta ,\; x_0 \leq x - \delta  \Rightarrow \mathcal{RS}_{stat} & = \frac{\exp(-\alpha(x-\delta-x_0)) - \exp(-\alpha(x-x_0))}{\exp(-\alpha(x-x_0)) - \exp(-\alpha(x+\delta-x_0))} \\
        & = \frac{\exp(\alpha \delta) - 1}{1 - \exp(-\alpha \delta)} 
    \end{split}
    \end{equation}
    \begin{equation}
    \begin{split}
        (x,\delta) : 0\leq \delta ,\; x + \delta \leq x_0  \Rightarrow \mathcal{RS}_{stat} & = \frac{\exp(\beta(x-\delta-x_0)) - \exp(\beta(x-x_0))}{\exp(\beta(x-x_0)) - \exp(\beta(x+\delta-x_0))} \\
        & = \frac{\exp(-\beta \delta) - 1}{1 - \exp(\beta \delta)} 
    \end{split}
    \end{equation}
    \begin{equation}
    \begin{split}
        (x,\delta) : x = x_0 ,\; 0 \leq \delta \Rightarrow \mathcal{RS}_{stat} & = \frac{\lambda \alpha}{\beta} \cdot \frac{1 - \exp(-\beta \delta)}{1 - \exp(-\alpha \delta)}.
    \end{split}
    \end{equation}
\end{subequations}
Considering that the Arrow-Pratt absolute risk aversion coefficient satisfies \(A(x)=\alpha\) for \(x_0 \leq x\) and \(A(x)=-\beta\) for \(x \leq x_0\), the exponential curvature parameters directly determine local risk sensitivity in each subdomain. In this specification, \(\alpha\) and \(\beta\) jointly determine both the curvature of the utility function (from nearly linear to highly steep or step-like behavior) and the corresponding probability thresholds.
In addition, because of the normalization terms in the denominators of the two branches, these parameters also affect the maximum attainable magnitude of the utility function within each subdomain. This creates an important limitation: within each branch, \(u(x)\) decreases as a function of \(\alpha\) or \(\beta\), so increasing curvature (through larger \(\alpha,\beta\)) simultaneously compresses the attainable range of value-function magnitudes. Hence, as the utility function becomes more step-like when \(\alpha\) and \(\beta\) increase, while risk sensitivity also changes across the gain and loss domains, the amplitude of the utility function is reduced. Moreover, \(\alpha\) and \(\beta\) play an additional role in shaping loss aversion.

\begin{figure}
    \centering
    \includegraphics[width=0.5\linewidth]{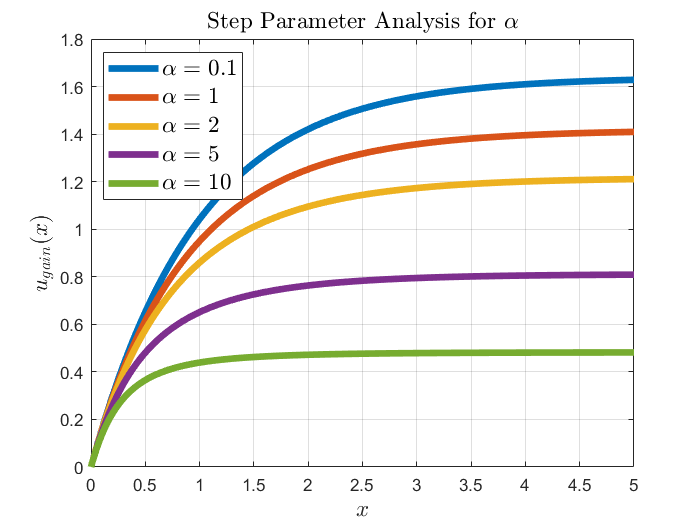}
   \caption{Effect of the parameter $\alpha$ on the gain branch of the utility function, for fixed values of the remaining parameters. Analogous behavior holds for the parameter $\beta$ in the loss subdomain.}
    \label{fig:alpha_parameter}
\end{figure}

A second limitation of the K\"{o}bberling--Wakker specification concerns the parameter $\lambda$, which acts as the scaling factor for loss aversion. This coefficient affects the left marginal slope at the reference point and the maximum attainable magnitude in the loss subdomain, without affecting probability thresholds when all outcomes lie entirely within the loss subdomain. However, $\lambda$ plays a crucial role in probabilistic risk sensitivity, particularly in determining probability thresholds, when outcomes span different subdomains, that is, when gains and losses interact. The limitation is that placing the loss-aversion coefficient exclusively in the loss branch restricts the flexibility of the specification, since it does not allow independent tuning of the maximum attainable value in the gain subdomain.

\begin{figure}
    \centering
    \includegraphics[width=0.5\linewidth]{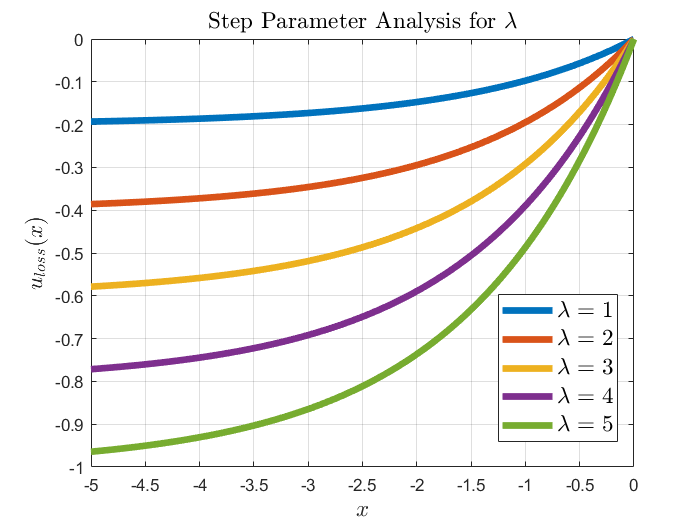}
    \caption{Effect of the parameter $\lambda$ on the loss branch of the utility function, for fixed values of the remaining parameters.}
    \label{fig:lambda_parameter}
\end{figure}

The third limitation of the K\"{o}bberling--Wakker specification is related to the constant absolute risk aversion (CARA) and constant probabilistic risk sensitivity (CPRS) properties imposed on each branch of the utility function. Since risk sensitivity is often closely tied to the reference (certain) state, these properties reduce the flexibility of the framework in capturing a broader range of heterogeneous risk preferences.

Another limitation concerns the influence of the domain on the risk sensitivity of each branch of the utility function. Existing formulations are typically defined on unbounded domains, which limits their applicability when outcomes are constrained to a bounded interval, as is often the case in real-world applications. Consider, for example, a domain restricted to $[-l_1, +l_2]$, with $l_1, l_2 \in \mathbb{R}^+$. In this case, the admissible values of $\delta$ for symmetric gain-loss comparisons are restricted to the interval $[0, \min\{l_1, l_2\}]$. This restriction introduces a significant issue: along the direction of maximum asymmetry, namely, when comparisons involve shifts in the interval $(\min\{l_1, l_2\}, \max\{l_1, l_2\})$, the valuation may become imbalanced. In particular, the perceived value of the maximum gain may exceed the absolute value of the maximum loss, contradicting the standard behavioral implication of loss aversion and potentially leading to non-intuitive or undesirable model behavior. Furthermore, the evaluation of symmetric gambles becomes dependent on the specific domain of the utility function, undermining consistency in modeling \emph{symmetric bet aversion}. Typically, symmetry is assumed in the outcome space, with gambles of the form $\{x-\delta, 1-p; x+\delta, p\}$, where both $x$ and $\delta$ may lie within the same subdomain (gains or losses) or span across both. However, in many practical and behavioral contexts, symmetry is interpreted differently.

\begin{figure}
    \centering
    \includegraphics[width=0.5\linewidth]{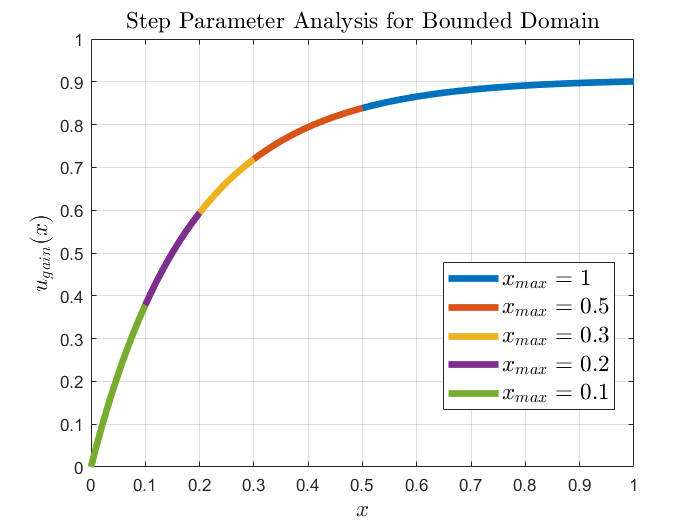}
    \caption{Effect of a bounded domain on the value-function specification, with the remaining parameters held fixed.}
    \label{fig:domain_influence}
\end{figure}

The final limitation of the K\"{o}bberling--Wakker specification is its incompatibility with \emph{Neilson’s aversion} definitions. Specifically, the exponential utility function exhibits a saturation effect as $x \to -\infty$, causing the marginal value of large losses to approach zero. This behavior violates Neilson-type weak-aversion requirements, which rule out loss-tail saturation by requiring sufficiently large relative loss valuation even at extreme outcome levels.

\subsection{An Extension Compatible with Neilson-type Aversion}
In this subsection, we develop a step-by-step extension of the proposed K\"{o}bberling--Wakker utility function. 

First, we extend the exponential-based CPT-utility function by introducing distinct parameters in the exponential coefficients and in the normalization terms, and by allowing a separate scaling factor in the gain branch, as follows:
\begin{equation} \label{eq:Extension_KW_v1}
    u^{(1)}(x) = \left\{
        \begin{array}{ll}
            \lambda_1\dfrac{1 - \exp\left(\dfrac{\alpha}{\gamma_1} (x - x_0)  \right)}{\alpha} & \text{for }x \geq x_0 \\
            \lambda_2 \dfrac{1 - \exp\left(\dfrac{\beta}{\gamma_2} (x - x_0)  \right)}{\beta} & \text{for } x < x_0 \\
        \end{array} 
    \right. 
\end{equation}
To begin specifying the parameter values, the utility function must be strictly increasing on the whole domain. This requirement leads to the following parameter restrictions:
\begin{equation}
    \frac{\lambda_1}{\gamma_1} < 0 \text{ and } \frac{\lambda_2}{\gamma_2} < 0
\end{equation}
By imposing monotonicity, we ensure that $u_+(x) > 0 \; \forall x > x_0$ and $u_-(x) < 0 \; \forall x < x_0$. To characterize the flexibility of the utility function in capturing the classical CPT S-shape, as well as other possible shapes, we proceed step by step and analyze each subdomain separately.
\begin{itemize}
    \item \underline{Gain subdomain}:
        \begin{itemize}
            \item \underline{Linear}: By taking $\alpha \to 0$ and applying L'H\^opital's rule to resolve the indeterminate form, we obtain
            \[
                u^{(1)}_+(x) = - \frac{\lambda_1}{\gamma_1} \cdot (x-x_0) 
            \]
            \item \underline{Convex}: By taking the second derivative with respect to $x$, we obtain the following condition on the parameters:
            \begin{equation}
                \text{Convex in gain subdomain: } \frac{\alpha}{\gamma_1} > 0
            \end{equation}
            
             \item \underline{Concave}: By taking the second derivative with respect to $x$, we obtain the following conditions on the parameters:
            \begin{equation}
                \text{Concave in gain subdomain: }  \frac{\alpha}{\gamma_1} < 0
            \end{equation}
        \end{itemize}
    \item \underline{Loss subdomain}:
         \begin{itemize}
            \item \underline{Linear}: By taking $\beta \to 0$ and applying L'H\^opital's rule to resolve the indeterminate form, we obtain
            \[
                u^{(1)}_-(x) = - \frac{\lambda_2}{\gamma_2} \cdot (x-x_0) 
            \]
            \item \underline{Convex}: By taking the second derivative with respect to $x$, we obtain the following condition on the parameters:
            \begin{equation}
                \text{Convex in loss subdomain: } \frac{\beta}{\gamma_2} > 0
            \end{equation}
            \item \underline{Concave}: By taking the second derivative with respect to $x$, we obtain the following condition on the parameters:
            \begin{equation}
                \text{Concave in loss subdomain:} \frac{\beta}{\gamma_2} < 0
            \end{equation}
        \end{itemize}
\end{itemize}

Proceeding to the effect of the parameters $\alpha$, $\beta$, $\gamma_1$, and $\gamma_2$ on the curvature of the utility function, and to the reasoning behind how this modification addresses some limitations of the classical framework, we first note that the previously mentioned limitation concerns the S-shaped value-function specification. Specifically, as $\alpha \to 0$, the gain branch of the utility function approaches a linear form. Conversely, as $\alpha \to +\infty$ with $\frac{\alpha}{\gamma_1} < 0$ in the gain subdomain, or as $\beta \to +\infty$ with $\frac{\beta}{\gamma_2} > 0$ in the loss subdomain, the utility function approximates a step function (see Figure~\ref{fig:alpha_parameter} for the gain subdomain; analogous behavior holds for the loss subdomain). 

In contrast, the parameters $\gamma_1$ and $\gamma_2$ exhibit the opposite effect under the same sign conditions (i.e., $\frac{\alpha}{\gamma_1} < 0$ in the gain subdomain and $\frac{\beta}{\gamma_2} > 0$ in the loss subdomain). For small values of $\gamma_1$ and $\gamma_2$, the utility function resembles a step function, whereas for large values it approaches a linear form while allowing independent control of the amplitude (see Figure~\ref{fig:gamma_parameter} for the gain subdomain; analogous behavior holds for the loss subdomain). In summary, by appropriately tuning the parameter pairs $\alpha$--$\gamma_1$ and $\beta$--$\gamma_2$, it is possible to recover the desired S-shape of the utility function while mitigating the undesirable effect of amplitude reduction.
\begin{figure}
    \centering
    \includegraphics[width=0.5\linewidth]{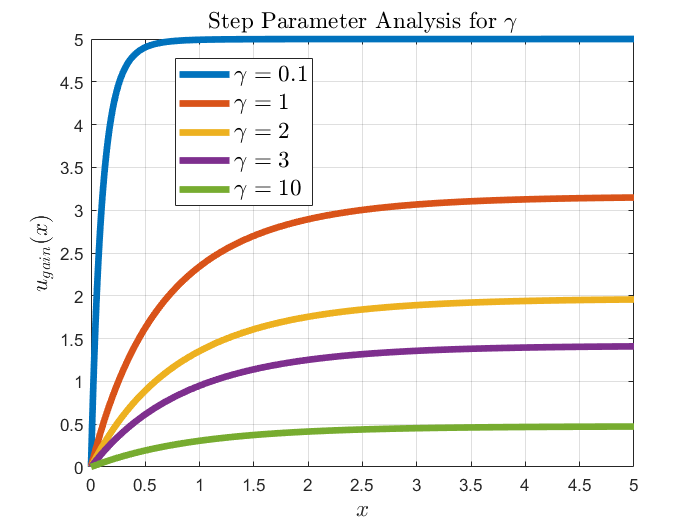}
    \caption{Effect of the parameter $\gamma_1$ on the gain branch of the utility function, for fixed values of the remaining parameters. Analogous behavior holds for the parameter $\gamma_2$ in the loss subdomain.}
    \label{fig:gamma_parameter}
\end{figure}

Moreover, the probabilistic risk sensitivity metric for a symmetric binary gamble around a certain outcome (stationary state) is modified as follows for the proposed utility function:
\begin{subequations}
    \begin{equation}
    \begin{split}
        (x,\delta) : 0 \leq \delta ,\; x_0 \leq x - \delta  \Rightarrow \mathcal{RS}_{stat} 
        & = \frac{\exp\left(-\frac{\alpha}{\gamma_1}\delta\right) - 1}{1 - \exp\left(\frac{\alpha}{\gamma_1}\delta\right)}
    \end{split}
    \end{equation}
    \begin{equation}
    \begin{split}
        (x,\delta) : 0\leq \delta ,\; x + \delta \leq x_0  \Rightarrow \mathcal{RS}_{stat}
        & = \frac{\exp\left(-\frac{\beta}{\gamma_2}\delta\right) - 1}{1 - \exp\left(\frac{\beta}{\gamma_2} \delta\right)}
    \end{split}
    \end{equation}
    \begin{equation}
    \begin{split}
        (x,\delta) : x = x_0 ,\; 0 \leq \delta \Rightarrow \mathcal{RS}_{stat}
        & = - \frac{\lambda_2 \alpha}{\lambda_1 \beta} \cdot \frac{1 - \exp\left(-\frac{\beta}{\gamma_2} \delta\right)}{1 - \exp\left(\frac{\alpha}{\gamma_1} \delta\right)}
    \end{split}
    \end{equation}
\end{subequations}
The Arrow-Pratt absolute risk aversion coefficient satisfies $A(x) = \frac{\alpha}{\gamma_1} \text{for} x_0 \leq x$ and $A(x) = -\frac{\beta}{\gamma_2} \text{for} x \leq x_0$. As we can see, relative to the K\"{o}bberling--Wakker utility function, the risk-sensitivity metrics differ only through the coefficient of the exponential term, which determines the curvature of the utility function in each subdomain.

Continuing with the effect of introducing a loss-aversion scaling factor in the gain subdomain, we show through $\mathcal{RS}^{equal}_{gamble}$ that the probability threshold depends only on the ratio $\frac{\lambda_2}{\lambda_1}$, and hence $\lambda_1$ determines the maximum absolute value in the gain subdomain. To begin, we divide the analysis into the following parts:
\begin{itemize}
    \item $x + \delta_4 < x_0$ or $x_0 < x - \delta_1$: The probability ratio threshold is given by 
    \begin{subequations}
        \begin{equation}
            \mathcal{RS}_{gamble}^{equal} = \frac{\exp\left(-\frac{\alpha}{\gamma_1} \delta_1)\right) - \exp\left(-\frac{\alpha}{\gamma_1} \delta_3\right)}{\exp\left(\frac{\alpha}{\gamma_1} \delta_2\right) - \exp\left(\frac{\alpha}{\gamma_1} \delta_4\right)} \;,\; x_0 < x - \delta_1
        \end{equation}
        \begin{equation}
            \mathcal{RS}_{gamble}^{equal} = \frac{\exp\left(-\frac{\beta}{\gamma_2} \delta_1\right) - \exp\left(-\frac{\beta}{\gamma_2} \delta_3\right)}{\exp\left(\frac{\beta}{\gamma_2} \delta_2\right) - \exp\left(\frac{\beta}{\gamma_2} \delta_4\right)} \;,\; x + \delta_4 < x_0
        \end{equation}
    \end{subequations}    
    \item $\{x_0 < x +\delta_4,x + \delta_2 < x_0\}$: The probability ratio threshold is given by
    \begin{equation}
        \mathcal{RS}_{gamble}^{equal} = \frac{\exp\left(-\frac{\beta}{\gamma_2} \delta_1\right) - \exp\left(-\frac{\beta}{\gamma_2} \delta_3\right)}{-1+\exp\left(\frac{\beta}{\gamma_2} \delta_2\right) + \frac{\lambda_1 \beta}{\lambda_2 \alpha} \left(1-\exp\left(\frac{\alpha}{\gamma_1} \delta_4\right)\right)}
    \end{equation}
    \item $\{x_0 < x + \delta_2, x - \delta_3 < x_0\}$: The probability ratio threshold is equal to
    \begin{equation}
        \mathcal{RS}_{gamble}^{equal} = \frac{\lambda_2 \alpha}{\lambda_1 \beta} \frac{\exp\left(-\frac{\beta}{\gamma_2} \delta_1\right) - \exp\left(-\frac{\beta}{\gamma_2} \delta_3\right)}{\exp\left(\frac{\alpha}{\gamma_1} \delta_2\right) - \exp\left(\frac{\alpha}{\gamma_1} \delta_4\right)}
    \end{equation}
    \item $\{x_0 < x - \delta_3, x - \delta_1 < x_0\}$: The probability ratio threshold is equal to
    \begin{equation}
        \mathcal{RS}_{gamble}^{equal} = \frac{-\frac{\lambda_2 \alpha}{\lambda_1 \beta}\left(1-\exp\left(-\frac{\beta}{\gamma_2} \delta_1\right)\right) + 1 - \exp\left(-\frac{\alpha}{\gamma_1} \delta_3\right)}{\exp\left(\frac{\alpha}{\gamma_1} \delta_2\right) - \exp\left(\frac{\alpha}{\gamma_1} \delta_4\right)}
    \end{equation}
\end{itemize}

Next, we introduce a new function \(g(x,x_0)\) in each subdomain, defined as a composition with the function \(u(x)\) in Equation~\ref{eq:Extension_KW_v1}, in order to incorporate the effect of the reference (certain) state on the probabilistic risk-sensitivity metric under the CPT applications and in order to give greater flexibility in the definition of curvature, rate at reference point and maximum value under the more general utility-theoretic applications, especially in the cases of bounded domains.
\begin{equation}\label{eq:Extension_KW_v2}
u^{(2)}(x)=
\begin{cases}
\lambda_1 \dfrac{1-\exp\!\left(\dfrac{\alpha}{\gamma_1}g_1(x,x_0)\right)}{\alpha}, & x \ge x_0, \\[6pt]
\lambda_2 \dfrac{1-\exp\!\left(\dfrac{\beta}{\gamma_2}g_2(x,x_0)\right)}{\beta}, & x < x_0.
\end{cases}
\end{equation}
In order to preserve monotonicity, the curvature of each branch of \(u(x)\) in Equation~\ref{eq:Extension_KW_v1}, and the normalization at the reference point, the functions \(g_1(x,x_0)\) and \(g_2(x,x_0)\) must satisfy the following properties, taking into account the properties of composite functions:
\begin{subequations}
    \begin{equation}
        g_i(x_0,x_0)=0, \qquad i\in\{1,2\}
    \end{equation}
    \begin{equation}
        g_i(\,\cdot\,,x_0)\ \text{is strictly increasing in } x, \qquad i\in\{1,2\}
    \end{equation}
    \begin{equation}
    \begin{split}
        & \text{If } u^{(1)}_{\pm}\ \text{is affine, convex, or concave on its subdomain, then } \\
        & \quad \quad \quad \quad \quad g_i(\,\cdot\,,x_0)\ \text{has the corresponding shape}
    \end{split}
    \end{equation}
    \begin{equation}
        \text{If } u^{(2)}_{\pm}\ \text{is required to satisfy CARA or CPRS, then } g_i(\,\cdot\,,x_0)\ \text{must be affine}
    \end{equation}
\end{subequations}
The effect of the functions $g_{1}(x,x_0)$ and $g_{2}(x,x_0)$ on risk sensitivity within each subdomain can be seen in the expression for $\mathcal{RS}$:
\begin{subequations}
    \begin{equation}
    \begin{split}
        (x,\delta) : 0 \leq \delta,\; x_0 \leq x - \delta \Rightarrow \mathcal{RS}_{stat}
        &=
        \frac{\exp\left(\frac{\alpha}{\gamma_1}g_1(x-\delta,x_0)\right)-\exp\left(\frac{\alpha}{\gamma_1}g_1(x,x_0)\right)}
        {\exp\left(\frac{\alpha}{\gamma_1}g_1(x,x_0)\right)-\exp\left(\frac{\alpha}{\gamma_1}g_1(x+\delta,x_0)\right)}
    \end{split}
    \end{equation}
    \begin{equation}
    \begin{split}
        (x,\delta) : 0 \leq \delta,\; x + \delta \leq x_0 \Rightarrow \mathcal{RS}_{stat}
        &=
        \frac{\exp\left(\frac{\beta}{\gamma_2}g_2(x-\delta,x_0)\right)-\exp\left(\frac{\beta}{\gamma_2}g_2(x,x_0)\right)}
        {\exp\left(\frac{\beta}{\gamma_2}g_2(x,x_0)\right)-\exp\left(\frac{\beta}{\gamma_2}g_2(x+\delta,x_0)\right)}
    \end{split}
    \end{equation}
    \begin{equation}
    \begin{split}
        (x,\delta) : x = x_0,\; 0 \leq \delta \Rightarrow \mathcal{RS}_{stat}
        &=
        - \frac{\lambda_2 \alpha}{\lambda_1 \beta}
        \cdot
        \frac{1-\exp\left(\frac{\beta}{\gamma_2}g_2(x_0-\delta,x_0)\right)}
        {1-\exp\left(\frac{\alpha}{\gamma_1}g_1(x_0+\delta,x_0)\right)}
    \end{split}
    \end{equation}
\end{subequations}
In Figures~\ref{fig:RS_stat_KT_gain},\ref{fig:RS_stat_KT_loss},\ref{fig:RS_stat_extended_KW_gain},\ref{fig:RS_stat_extended_KW_loss}, we present the graphical representation of the effect of $g_{1,2}(x)$ on the risk-sensitivity metric for the case of an S-shaped utility function, comparing the behavior of the Kahneman--Tversky approach with our modified K\"{o}bberling--Wakker specification. As we can see, the utility function induces a stronger expression of the risk behavior within each subdomain as we move toward the reference point in both approaches.
\begin{figure}[!h]
\centering
\begin{subfigure}{.5\textwidth}
    \centering
    \includegraphics[width=\textwidth]{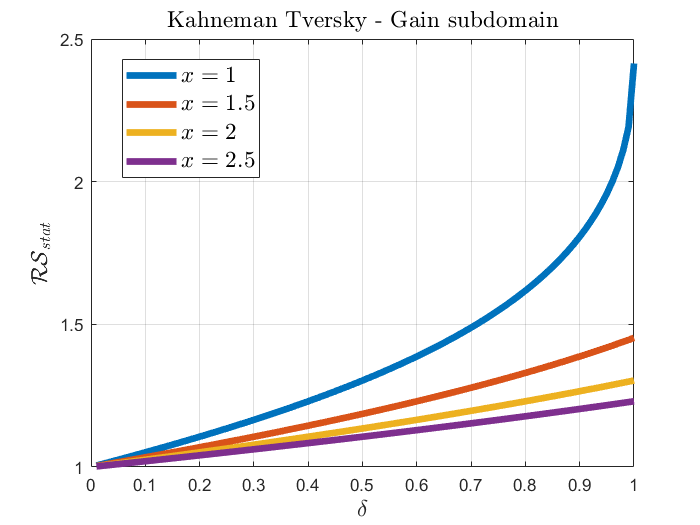}
    \caption{$u_+(x)=x^{0.5}$}
    \label{fig:RS_stat_KT_gain}
\end{subfigure}%
\begin{subfigure}{.5\textwidth}
    \centering
    \includegraphics[width=\textwidth]{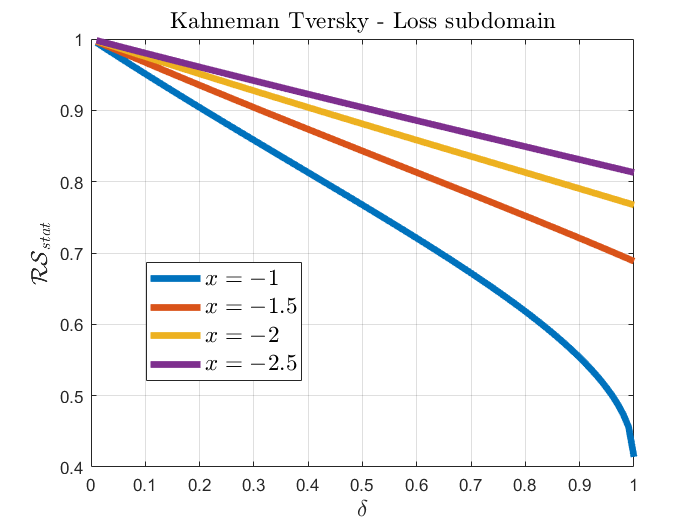}
    \caption{$u_-(x)=-2(-x)^{0.5}$}
    \label{fig:RS_stat_KT_loss}
\end{subfigure}
\begin{subfigure}{.5\textwidth}
    \centering
    \includegraphics[width=\textwidth]{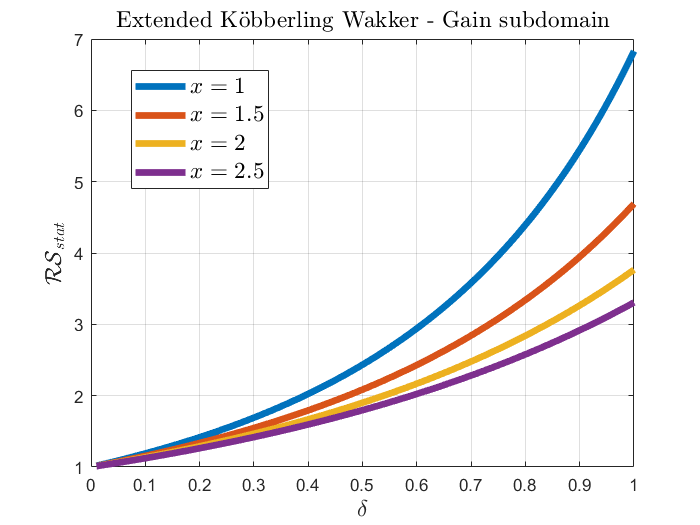}
    \caption{$u_+(x)=\dfrac{1-\exp\!\left(-2(1-\exp(-x))\right)}{3}$}
    \label{fig:RS_stat_extended_KW_gain}
\end{subfigure}%
\begin{subfigure}{.5\textwidth}
    \centering
    \includegraphics[width=\textwidth]{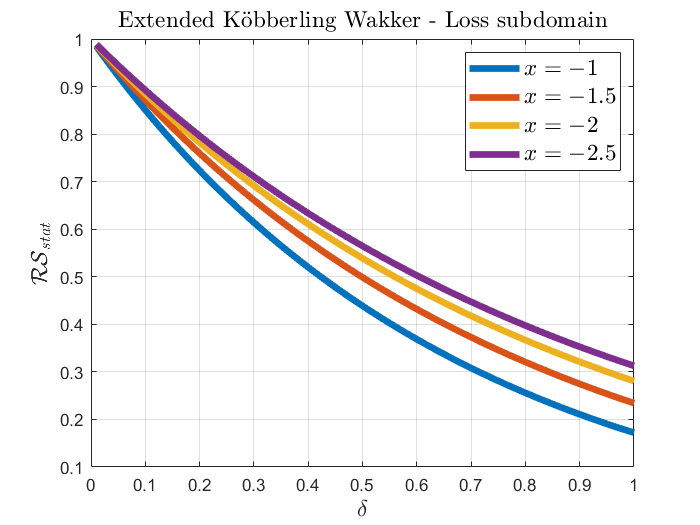}
    \caption{$u_-(x)=-2\,\dfrac{1-\exp\!\left(\frac{5}{3}(\exp(x)-1)\right)}{2}$}
    \label{fig:RS_stat_extended_KW_loss}
\end{subfigure}
\caption[Risk-sensitivity comparison with Kahneman--Tversky]{Graphical representation of the effect of $g_{1,2}(x)$ on the risk-sensitivity metric: comparison of the trend under the Kahneman--Tversky specification and the extended K\"{o}bberling--Wakker specification.}
\label{fig:RS_stat_comparison}
\end{figure}

The specification in Equation~\ref{eq:Extension_KW_v2} still cannot, by itself, accommodate Neilson-type aversion definitions for an S-shaped utility function. Indeed, we divide the analysis into three cases, taking into consideration the definition of \emph{Neilson's strong aversion}:
\begin{itemize}
    \item If $g_2(x)$ has a finite saturation effect, meaning that $\lim_{x \to -\infty} g_2(x)=c \in \mathbb{R}_{-}^{*}$ and \newline $\lim_{x \to -\infty} g_2'(x)=0$, then $\lim_{x \to -\infty} \frac{\partial u_-^{(2)}(x)}{\partial x}=0$.

    \item If $g_2(x)$ has an infinite saturation effect, meaning that $\lim_{x \to -\infty} g_2(x)=-\infty$ and \newline $\lim_{x \to -\infty} g_2'(x)=0$, then $\lim_{x \to -\infty} \frac{\partial u_-^{(2)}(x)}{\partial x}=0$.

    \item If $g_2(x)$ has asymptotically linear behavior at $-\infty$, meaning that $\lim_{x \to -\infty} g_2(x)=-\infty$ and $\lim_{x \to -\infty} g_2'(x)=c \in \mathbb{R}_{+}^{*}$, then $\lim_{x \to -\infty} \frac{\partial u_-^{(2)}(x)}{\partial x}=0$.
\end{itemize}
Hence, we introduce an additional component, inspired by the equivalent characterization of \emph{Neilson's weak aversion}, in order to provide greater flexibility in the structure of the proposed utility function, as follows:
\begin{equation}\label{eq:Extension_KW_v3}
u^{(3)}(x)=
\begin{cases}
\lambda_1 \dfrac{1-\exp\!\left(\dfrac{\alpha}{\gamma_1}g_1(x,x_0)\right)}{\alpha}, & x \ge x_0,\\[6pt]
\lambda_2 \dfrac{1-\exp\!\left(\dfrac{\beta}{\gamma_2}g_2(x,x_0)\right)}{\beta}+\lambda_3(x-x_0), & x < x_0.
\end{cases}
\end{equation}
\begin{proposition}[Tail-slope compatibility with Neilson-type aversion]\label{prop:tail_slope_v3}
Assume that $u^{(3)}$ is strictly increasing and S-shaped on the gain and loss subdomains, with $u^{(3)}(x_0)=0$. If
$\left. \frac{\partial u^{(2)}_+(x)}{\partial x} \right\vert _{x=x_0} \leq \lambda_3$,
then the linear loss-tail component in Equation~\ref{eq:Extension_KW_v3} rules out loss-side saturation and provides the tail-slope condition required for Neilson-type weak-to-strong compatibility. In particular, under the regularity assumptions of Theorem~\ref{Neilson_S_shape_equivalence}, the resulting specification satisfies Neilson's strong aversion whenever it satisfies Neilson's weak aversion.
\end{proposition}
\begin{proof}
The added term $\lambda_3(x-x_0)$ makes the limiting loss-domain marginal value bounded below by $\lambda_3$. The stated inequality therefore ensures that the loss-tail marginal value is no smaller than the right marginal value at the reference point, which is precisely the no-saturation condition in Theorem~\ref{Neilson_S_shape_equivalence}.
\end{proof}

\subsection{Evaluation of Selected Characteristic Utility Function Scenarios}
To better understand how different decision-maker profiles influence decision-making under uncertainty, we examine a series of representative utility function cases.

First, we present the \textbf{CPT benchmark} profile, shown in Figure~\ref{Restrained_Profiteer}. This profile exhibits risk-seeking behavior in the loss subdomain and risk-averse behavior in the gain subdomain. In practice, this translates to a tendency to move toward the reference point $x_0$ when facing losses (i.e., a preference for fair symmetric binary gambles involving potential losses), and a more cautious approach in the gain domain, preferring symmetric gambles only when outcomes lie within the gain subdomain and involve low uncertainty. This type of behavior across both subdomains aligns with the classical assumptions of CPT.

Second, the \textbf{loss-conservative/gain-seeking} decision-maker profile is illustrated in Figure~\ref{Defeatist_Greedy}. This decision maker displays a cautious attitude in the loss subdomain, requiring a high degree of certainty about a profitable outcome before engaging in a symmetric binary gamble. In contrast, in the gain subdomain, the decision maker demonstrates strong mobility toward higher rewards, indicating a willingness to accept greater uncertainty when pursuing larger gains. This asymmetric attitude toward risk characterizes the combination of loss conservatism and gain seeking.

Another possible decision-maker profile is the \textbf{globally risk-seeking} profile, shown in Figure~\ref{Greedy_Profiteer}. In this case, the decision maker exhibits risk-seeking behavior in both subdomains: a tendency to move toward the reference point $x_0$ in the loss domain, and toward $+\infty$ in the gain domain. This reflects a consistent preference for high-risk, high-reward gambles, regardless of whether the outcomes lie in the loss or gain subdomain.

By contrast, a decision maker who exhibits risk aversion in both subdomains can be classified as a \textbf{globally risk-averse} profile, as shown in Figure~\ref{Defeatist}. This decision maker prefers certainty and avoids engaging in symmetric binary gambles unless the outcomes are highly predictable, regardless of whether they fall in the loss or gain subdomain.

Additionally, some decision makers may adopt a purely rational decision-making approach, with minimal behavioral distortions. This profile, illustrated in Figure~\ref{Non_Behavioral_Loss_Averse}, represents the \textbf{non-behavioral profile with loss aversion}. Such a decision maker is risk-neutral across both subdomains and may display only minimal loss aversion. Their utility function reflects a linear, unmodulated response to outcomes, in contrast to profiles influenced by psychological or other subjective factors.

\begin{figure}[!t]
\centering
\begin{subfigure}{.33\textwidth}
    \centering
    \includegraphics[width=\textwidth]{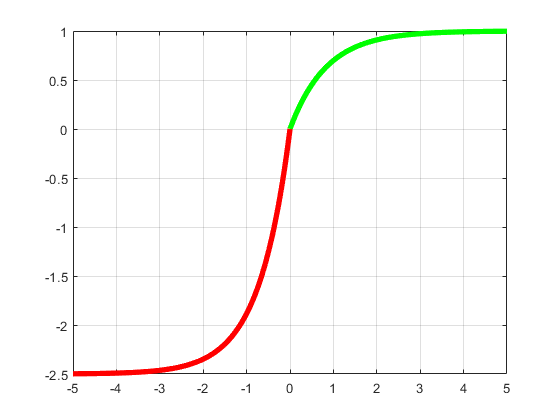}
    \caption{CPT benchmark}
    \label{Restrained_Profiteer}
\end{subfigure}%
\begin{subfigure}{.33\textwidth}
    \centering
    \includegraphics[width=\textwidth]{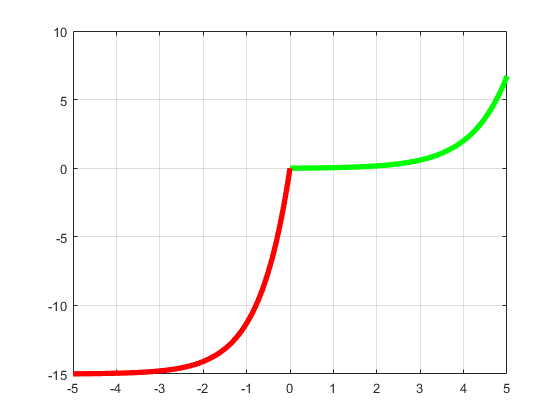}
    \caption{Globally risk-seeking}
    \label{Greedy_Profiteer}
\end{subfigure}%
\begin{subfigure}{.33\textwidth}
    \centering
    \includegraphics[width=\textwidth]{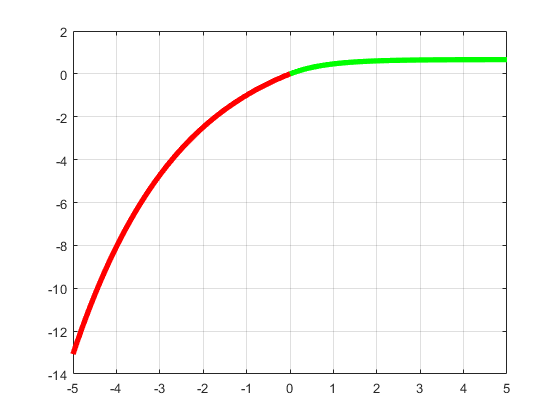}
    \caption{Globally risk-averse}
    \label{Defeatist}
\end{subfigure}
\begin{subfigure}{.33\textwidth}
    \centering
    \includegraphics[width=\textwidth]{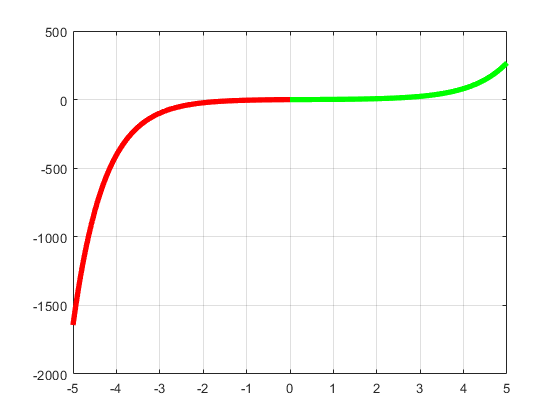}
    \caption{Loss-conservative/gain-seeking}
    \label{Defeatist_Greedy}
\end{subfigure}
\begin{subfigure}{.33\textwidth}
    \centering
    \includegraphics[width=\textwidth]{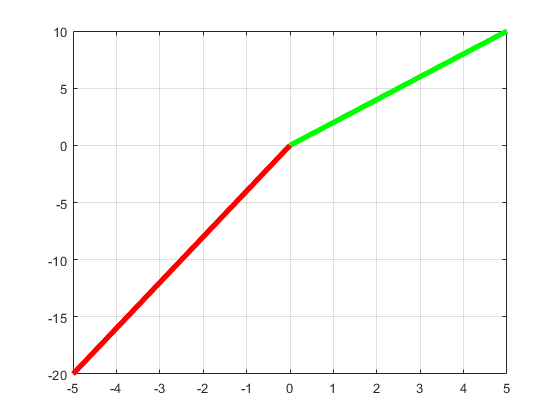}
    \caption{Non-Behavioral with loss aversion}
    \label{Non_Behavioral_Loss_Averse}
\end{subfigure}
\caption[Illustrative decision-maker profiles]{Illustrative decision-maker profiles under the proposed utility-function specification.}
\end{figure}

\section{Conclusion}\label{Section5}

This work provides a rigorous analysis and a formal bridge between \emph{Neilson's weak} and \emph{strong aversion}, revealing the conditions under which these notions converge or diverge under a nonsmooth S-shaped value-function specification. Furthermore, we provide a gamble-based interpretation of \emph{Neilson's strong aversion}.

Moreover, we identify inherent limitations of the standard K\"{o}bberling--Wakker value-function specification, especially the CARA and CPRS properties, and the amplitude reduction induced by increasing the $\alpha$, $\beta$ parameters, and we clarify modeling issues that arise in application scenarios, such as the treatment of asymmetric outcomes under symmetric modeling assumptions and the handling of bounded domains. To address these limitations, we propose a modified version of the K\"{o}bberling--Wakker utility function that enhances the expressive capacity of CPT, allowing for a more accurate representation of individual behavior in complex decision environments.

The proposed generalized framework opens several promising research directions. First, it calls for empirical validation across diverse decision-making contexts, such as intertemporal choice, environmental risk, or AI-human interaction systems. Second, it offers potential applications in algorithmic design for decision-support systems that must account for nuanced and context-dependent human risk preferences. Finally, it raises foundational questions about the normative versus descriptive role of utility functions in behavioral modeling, especially as decision systems evolve to be more adaptive and personalized.

In conclusion, by strengthening the analytical foundations of CPT and extending its parametric flexibility, this work contributes to a more robust behavioral decision-theoretic framework for modeling judgment and choice under uncertainty.

\section*{Acknowledgment}
This work is supported by the European Research Council (ERC) under the EU’s Horizon 2020 research and innovation programme (Grant agreement No. 101003431). S. Vaidanis is partially supported by the Onassis Foundation (Scholarship ID: F ZU 076-1/2024-2025). 

\begin{appendices} \label{proofs}

\section{Proofs and Technical Preliminaries}
In this section, we present the proofs of the theorems referenced in the main text. Before proceeding, we introduce several useful propositions that serve as technical preliminaries for the subsequent proofs.

\subsection{Proof of Theorem \ref{increasing_non_symmetric}}
\textbf{(Direct proof)}:
Consider a fair non-symmetric binary gamble around the reference point $x_0$, with outcomes $x_0-w\delta_1$ and $x_0+w\delta_2$, where $\delta_1,\delta_2>0$ and $w\in\mathbb{R}_+$. The expected value of the gamble is
\[
\mathbb{E}[\text{gamble}]
=
\frac{\delta_2}{\delta_1+\delta_2}(x_0-w\delta_1)
+
\frac{\delta_1}{\delta_1+\delta_2}(x_0+w\delta_2)
=
x_0,
\]
which confirms that the gamble is fair around $x_0$.

Define the utility premium for this decision setup by
\[
\mathcal{U}(x_0,\delta_1,\delta_2,w)
:=
u(x_0)
-\frac{\delta_2}{\delta_1+\delta_2}\,u(x_0-w\delta_1)
-\frac{\delta_1}{\delta_1+\delta_2}\,u(x_0+w\delta_2).
\]

To express increasing aversion to scaled gambles, we require \(\mathcal{U}(x_0,\delta_1,\delta_2,w)\) to be a strictly increasing function of \(w\). Differentiating with respect to \(w\), we obtain
\[
\begin{aligned}
\frac{\partial}{\partial w}\mathcal{U}(x_0,\delta_1,\delta_2,w)
&=
-\frac{\delta_2}{\delta_1+\delta_2}(-\delta_1)\,u'(x_0-w\delta_1)
-\frac{\delta_1}{\delta_1+\delta_2}(\delta_2)\,u'(x_0+w\delta_2) \\
&=
\frac{\delta_1\delta_2}{\delta_1+\delta_2}
\left(u'(x_0-w\delta_1)-u'(x_0+w\delta_2)\right).
\end{aligned}
\]

Thus,
\[
\frac{\partial}{\partial w}\mathcal{U}(x_0,\delta_1,\delta_2,w)>0
\quad\Longleftrightarrow\quad
u'(x_0+w\delta_2)<u'(x_0-w\delta_1),
\qquad \forall\, w>0,
\]
which is the desired condition.

\textbf{(Converse proof)}:
Suppose that
\[
u'(x_0+w\delta_2)<u'(x_0-w\delta_1), \qquad \forall\, w>0.
\]
Then, from the derivative expression established above, it follows that
\[
\frac{\partial}{\partial w}\mathcal{U}(x_0,\delta_1,\delta_2,w)>0, \qquad \forall\, w>0.
\]
Hence, \(\mathcal{U}(x_0,\delta_1,\delta_2,w)\) is strictly increasing in \(w\), which implies increasing aversion to scaled versions of the fair non-symmetric gamble.

\subsection{Proof of Theorem \ref{weak_loss_aversion_equivalent}}
\textbf{(Direct proof)}:
\emph{Neilson's weak loss aversion} states that for any \(y<x_0<z\),
\[
\frac{u(z)}{z-x_0} \leq \frac{u(y)}{y-x_0}.
\]
Taking the limit as \(z\to x_0^+\), we obtain
\[
\lim_{z\to x_0^+}\frac{u(z)}{z-x_0}=u'(x_0^+),
\]
and therefore
\[
\frac{u(y)}{y-x_0} \geq u'(x_0^+), \qquad \forall\, y<x_0.
\]
Since \(y-x_0<0\), multiplying both sides by \(y-x_0\) reverses the inequality, yielding
\[
u(y) \leq u'(x_0^+)(y-x_0), \qquad \forall\, y<x_0,
\]
which is the claimed inequality.

\textbf{(Converse proof)}:
Suppose that
\[
u(y) \leq u'(x_0^+)(y-x_0), \qquad \forall\, y<x_0.
\]
We want to show that this implies \emph{Neilson's weak loss aversion}, i.e.,
\[
\frac{u(z)}{z-x_0} \leq \frac{u(y)}{y-x_0}, \qquad \forall\, y<x_0<z.
\]

By the Mean Value Theorem, for any \(z>x_0\), there exists \(\xi\in(x_0,z)\) such that
\[
\frac{u(z)-u(x_0)}{z-x_0}=u'(\xi).
\]
Since \(u(x_0)=0\), this becomes
\[
\frac{u(z)}{z-x_0}=u'(\xi).
\]

Because \(u\) is concave on the gain domain and strictly increasing, \(u'\) is decreasing on \((x_0,\infty)\). Hence,
\[
u'(\xi)\le u'(x_0^+).
\]
Under strict concavity, the inequality is strict:
\[
u'(\xi)<u'(x_0^+).
\]

On the other hand, the assumption \(u(y) \leq u'(x_0^+)(y-x_0)\) and the fact that \(y-x_0<0\) imply
\[
\frac{u(y)}{y-x_0} \geq u'(x_0^+).
\]

Combining the previous inequalities, we obtain
\[
\frac{u(z)}{z-x_0}=u'(\xi) \leq u'(x_0^+) \leq \frac{u(y)}{y-x_0},
\]
which proves \emph{Neilson’s weak loss aversion}. Therefore, the two statements are equivalent.

\subsection{Proof of Theorem \ref{Neilson_S_shape_equivalence}}
Assume that \emph{Neilson's weak loss aversion} holds. By Theorem~\ref{weak_loss_aversion_equivalent}, this is equivalent to
\[
u(x) \leq u'(x_0^+)(x-x_0), \qquad \forall x<x_0.
\]
Moreover, under the S-shaped assumption (convexity on the loss domain and concavity on the gain domain), we also have
\[
u'(x_0^-)(x-x_0) \leq u(x), \qquad \forall x<x_0,
\]
and therefore
\[
u'(x_0^-)(x-x_0) \leq u(x) \leq u'(x_0^+)(x-x_0), \qquad \forall x<x_0.
\]

Since \(x-x_0\to -\infty\) as \(x\to -\infty\), both bounding linear functions tend to \(-\infty\). Hence, by the squeeze theorem,
\[
\lim_{x\to -\infty}u(x)=-\infty.
\]

Assume, in addition, that the loss-domain marginal value satisfies \eqref{eq:NS_tail_slope}. Then, for every \(y<x_0\),
\[
u'(y)\geq \inf_{x<x_0}u'(x)\ge u'(x_0^+).
\]

On the gain domain, concavity implies that \(u'\) is decreasing on \((x_0,\infty)\), so
\[
u'(z)\leq u'(x_0^+), \quad \forall z>x_0.
\]
Combining the two inequalities, we obtain
\[
u'(z)\leq u'(x_0^+)\leq u'(y), \quad \forall y<x_0<z,
\]
which is exactly \emph{Neilson’s strong loss aversion}.

\subsection{Proof of Theorem \ref{Neilson_strong_equivalent}}
\textbf{(i) $\Rightarrow$ (ii).}
Assume that for all $y<x_0<z$,
\[
u'(z)\le u'(y).
\]
Fix $\delta_1,\delta_2,\delta_3,\delta_4\ge 0$ such that
$\delta_1\neq \delta_3$, $\delta_2\neq \delta_4$, and $\delta_1-\delta_3>0$, $\delta_4-\delta_2>0$.

Define
\[
a:=x_0+\delta_4,\quad b:=x_0+\delta_2,
\qquad
c:=x_0-\delta_1,\quad d:=x_0-\delta_3.
\]
Then $a\neq b$ and $c\neq d$. By the Mean Value Theorem applied to $u$ on the interval
between $a$ and $b$ (which lies in $[x_0,\infty)$), there exists $z_\ast$ between $a$ and $b$ such that
\[
\frac{u(b)-u(a)}{b-a}=u'(z_\ast).
\]
Equivalently,
\[
\frac{u(x_0+\delta_2)-u(x_0+\delta_4)}{\delta_2-\delta_4}=u'(z_\ast).
\]

Similarly, by the Mean Value Theorem applied to $u$ on the interval between $c$ and $d$
(which lies in $(-\infty,x_0]$), there exists $y_\ast$ between $c$ and $d$ such that
\[
\frac{u(d)-u(c)}{d-c}=u'(y_\ast).
\]
Since $d-c=(x_0-\delta_3)-(x_0-\delta_1)=\delta_1-\delta_3$, this is
\[
\frac{u(x_0-\delta_3)-u(x_0-\delta_1)}{\delta_1-\delta_3}=u'(y_\ast).
\]

Because $y_\ast<x_0<z_\ast$, the assumed non-strict strong loss aversion yields
$u'(z_\ast)\le u'(y_\ast)$. Substituting the two difference quotients gives (ii):
\[
\frac{u(x_0+\delta_2)-u(x_0+\delta_4)}{\delta_2-\delta_4}
\le
\frac{u(x_0-\delta_3)-u(x_0-\delta_1)}{\delta_1-\delta_3}.
\]

\medskip
\textbf{(ii) $\Rightarrow$ (i).}
Assume (ii) holds. Fix arbitrary $y<x_0<z$ and set
\[
\bar\delta_1:=x_0-y>0,\qquad \bar\delta_4:=z-x_0>0.
\]
Choose sequences $\delta_3^{(n)}\to \bar\delta_1$ and $\delta_2^{(n)}\to \bar\delta_4$
such that $\delta_3^{(n)}\neq \bar\delta_1$ and $\delta_2^{(n)}\neq \bar\delta_4$ for all $n$.
Apply (ii) with
\[
\delta_1=\bar\delta_1,\quad \delta_4=\bar\delta_4,\quad
\delta_3=\delta_3^{(n)},\quad \delta_2=\delta_2^{(n)}.
\]
Then
\[
\frac{u(x_0+\delta_2^{(n)})-u(x_0+\bar\delta_4)}{\delta_2^{(n)}-\bar\delta_4}
\le
\frac{u(x_0-\delta_3^{(n)})-u(x_0-\bar\delta_1)}{\bar\delta_1-\delta_3^{(n)}}.
\]
As $n\to\infty$, differentiability on each side implies the left-hand side converges to
$u'(x_0+\bar\delta_4)=u'(z)$ and the right-hand side converges to $u'(x_0-\bar\delta_1)=u'(y)$.
Hence,
\[
u'(z)\le u'(y).
\]
Since $y<x_0<z$ were arbitrary, (i) holds.

\end{appendices}

\bibliographystyle{IEEEtran}
\bibliography{my_bibliography}

\end{document}